\documentstyle[jair,twoside,11pt,theapa]{article}
\input{amssymb.sty}
\jairheading{25}{2006}{315--348}{08/05}{03/06}
\firstpageno{315}

\title{Negotiating Socially Optimal Allocations of Resources}
\author{%
\name Ulle Endriss \email ulle@illc.uva.nl \\
\addr ILLC, University of Amsterdam \\
1018 TV Amsterdam, The Netherlands
\AND
\name Nicolas Maudet \email maudet@lamsade.dauphine.fr \\
\addr LAMSADE, Universit\'e Paris-Dauphine \\
75775 Paris Cedex 16, France
\AND
\name Fariba Sadri \email fs@doc.ic.ac.uk \\
\addr Department of Computing, Imperial College London \\
London SW7 2AZ, UK
\AND
\name Francesca Toni \email ft@doc.ic.ac.uk \\
\addr Department of Computing, Imperial College London \\
London SW7 2AZ, UK}

\ShortHeadings{Negotiating Socially Optimal Allocations of Resources}%
{Endriss, Maudet, Sadri, \& Toni}

\renewcommand{\emptyset}{\{\,\}}
\renewcommand{\min}{\mbox{\textit{min}}}
\renewcommand{\max}{\mbox{\textit{max}}}
\newcommand{\R}{\mathbb{R}}
\newcommand{\ie}{\textit{i.e.}}
\newcommand{\eg}{e.g.}

\newtheorem{definition}{Definition}
\newtheorem{theorem}{Theorem}
\newtheorem{lemma}{Lemma}

\newcommand{\proofbegin}{\noindent\emph{Proof.}}
\newcommand{\proofend}{\mbox{}~\hfill~$\Box$}
\newenvironment{proof}{\proofbegin}{\proofend}
\newcommand{\lefttorightproof}{`$\Rightarrow$':}
\newcommand{\righttoleftproof}{`$\Leftarrow$':}
\newcommand{\afterproof}{\paragraph*{$\!\!\!\!\!\!$}}

\begin{document}
\maketitle

\begin{abstract}\noindent
A multiagent system may be thought of as an artificial society of
autonomous software agents and we can apply concepts borrowed from  
welfare economics and social choice theory to assess the social welfare 
of such an agent society. In this paper, we study an abstract negotiation 
framework where agents can agree on multilateral deals to exchange bundles 
of indivisible resources. We then analyse how these deals affect social 
welfare for different instances of the basic framework and different 
interpretations of the concept of social welfare itself.
In particular, we show how certain classes of deals are both 
sufficient and necessary to guarantee that a socially optimal 
allocation of resources will be reached eventually. 
\end{abstract}

\section{Introduction}

A multiagent system may be thought of as an artificial society of
autonomous software agents. Negotiation over the distribution of 
resources (or tasks) amongst the agents inhabiting such a society
is an important area of research in artificial intelligence and computer 
science \cite{RosenscheinZlotkin1994a,Kraus2001a,ChavezEtAl1997,SandholmChapter1999}. 
A number of variants of this problem have been studied in the 
literature. Here we consider the case of an artificial society of agents 
where, to begin with, each agent holds a bundle of 
indivisible resources to which it assigns a certain utility. Agents may then 
negotiate with each other in order to agree on the redistribution of some of these 
resources to benefit either themselves or the agent society they inhabit.

Rather than being concerned with specific strategies for negotiation, 
we analyse how the redistribution of resources by means of negotiation affects the 
well-being of the agent society as a whole. To this end, we make use of formal tools
for measuring \emph{social welfare} developed in welfare economics and social choice 
theory \cite{Moulin1988a,social-choice-handbook}. In the multiagent systems literature, 
the \emph{utilitarian} interpretation of the concept of social welfare is usually taken 
for granted \cite{RosenscheinZlotkin1994a,SandholmChapter1999,Wooldridge2002a}, \ie\
whatever increases the average welfare of the agents inhabiting a society
is taken to be beneficial for society as well. 
This is not the case in welfare economics, for instance,
where different notions of social welfare are being studied and compared
with each other. Here, the concept of \emph{egalitarian} social welfare takes a
particularly prominent role \cite{Sen1970,Rawls1971a,Moulin1988a,social-choice-handbook}.
In this model, social welfare is tied to the individual welfare of the weakest member of society,
which facilitates the incorporation of a notion of fairness into the resource allocation process.
While the discussion of the respective advantages and drawbacks of different notions 
of social welfare in the social sciences tends to be dominated by ethical considerations,\footnote{%
A famous example is Rawls' \emph{veil of ignorance}, a thought experiment designed to 
establish what constitutes a \emph{just} society \cite{Rawls1971a}.} 
in the context of societies of artificial software agents the choice of a suitable formal 
tool for modelling social welfare boils down to a clear-cut (albeit not necessarily simple) 
technical design decision \cite{EndrissMaudetESAW2003}. 
Indeed, different applications may call for different social criteria.
For instance, for the application studied by \citeA{LemaitreEtAlIJCAI1999}, 
where agents need to agree on the access to an earth observation satellite which has been funded 
jointly by the owners of these agents, it is important that each one of them receives a ``fair'' share
of the common resource. Here, a society governed by egalitarian principles may be the most appropriate.
In an electronic commerce application running on the Internet where agents have little or
no commitments towards each other, on the other hand, egalitarian principles seem of
little relevance. In such a scenario, utilitarian social welfare would provide 
an appropriate reflection of the overall profit generated.
Besides utilitarian and egalitarian social welfare, we are also going to discuss  
notions such as \emph{Pareto} and \emph{Lorenz optimality} \cite{Moulin1988a}, 
as well as \emph{envy-freeness} \cite{BramsTaylor1996a}. 

In this paper, we study the effect that negotiation over resources has on society
for a number of different interpretations of the concept of social welfare.
In particular, we show how certain classes of deals regarding the exchange of resources
allow us to guarantee that a socially optimal allocation of resources will be reached 
eventually. These \emph{convergence results} may be interpreted as the emergence of a particular 
global behaviour (at the level of society) in reaction to local behaviour governed by the 
negotiation strategies of individual agents (which determine the kinds of deals 
agents are prepared to accept). 
The work described here is complementary to the large body of literature
on mechanism design and game-theoretical models of negotiation in multiagent
systems \cite<see e.g.>{RosenscheinZlotkin1994a,Kraus2001a,FatimaEtAlAIJ2004}.
While such work is typically concerned with negotiation at the local level
(how can we design mechanisms that provide an incentive to individual agents
to adopt a certain negotiation strategy?), we address negotiation at a global 
level by analysing how the actions taken by agents locally affect the overall 
system from a social point of view.

As we shall see, truly \emph{multilateral} deals 
involving any number of agents as well as any number of resources may be 
necessary to be able to negotiate socially optimal allocations of resources. 
This is certainly true as long as we use arbitrary utility functions to model
the preferences of individual agents. In some application domains, however, 
where utility functions may be assumed to be subject to certain restrictions
(such as being additive), we are able to obtain stronger results and show
that also structurally simpler classes of deals (in particular, deals involving 
only a single resource at a time) can be sufficient to negotiate socially 
optimal allocations. 
Nevertheless, for other seemingly strong restrictions on agents' utility functions
(such as the restriction to dichotomous preferences) we are able to show that no 
reduction in the structural complexity of negotiation is possible.

Our approach to multiagent resource allocation is of a \emph{distributed} nature.
In general, the allocation procedure used to find a suitable allocation of
resources could be either centralised or distributed.
In the centralised case, a single entity decides on the final
allocation of resources amongst agents, possibly after having elicited
the agents' preferences over alternative allocations. Typical
examples are combinatorial auctions \cite{CramtonEtAl2006}. Here the central 
entity is the auctioneer and the reporting of preferences takes the form of bidding.
In truly distributed approaches, on the other hand, allocations emerge
as the result of a sequence of local negotiation steps. 
Both approaches have their advantages and disadvantages.
Possibly the most important argument in favour of auction-based mechanisms
concerns the simplicity of the communication protocols required to implement
such mechanisms. Another reason for the popularity of centralised 
mechanisms is the recent push in the design of powerful algorithms for 
combinatorial auctions that, for the first time, perform reasonably well in 
practice \cite{FujishimaEtAlIJCAI1999,SandholmAIJ2002}. 
Of course, such techniques are, in principle, also applicable in the distributed
case, but research in this area has not yet reached the same level of maturity 
as for combinatorial auctions. 
An important argument \emph{against} centralised approaches is that it may be difficult 
to find an agent that could assume the role of an ``auctioneer'' (for instance, in
view of its computational capabilities or in view of its trustworthiness).

The line of research pursued in this paper has been inspired 
by Sandholm's work on sufficient and necessary contract (\ie\
deal) types for distributed task allocation \cite{Sandholm1998a}.
Since then, it has been further developed by the present authors, 
their colleagues, and others in the context of resource allocation
problems \cite{BouveretLangIJCAI2005,ChevaleyreEtAlCSDT2004,ChevaleyreEtAlAAMAS2005,%
ChevaleyreEtAlIJCAI2005,DunneJAIR2005,DunneEtAlECAI2004,DunneEtAlAIJ,EndrissMaudetESAW2003,%
EndrissMaudetJAAMAS2005,EndrissEtAlAAMAS2003-optimal,EndrissEtAlMFI2003}. 
In particular, we have extended Sandholm's framework by also addressing negotiation 
systems without compensatory side payments \cite{EndrissEtAlAAMAS2003-optimal}, as well 
as agent societies where the concept of social welfare is given a different interpretation 
to that in the utilitarian programme \cite{EndrissEtAlMFI2003,EndrissMaudetESAW2003}.
The present paper provides a comprehensive overview of the most fundamental
results, mostly on the convergence to an optimal allocation with respect
to different notions of social welfare, in a very active and timely area 
of ongoing research. 

The remainder of this paper is organised as follows.
Section~\ref{sec:prelim} introduces the basic negotiation framework
for resource reallocation we are going to consider. It gives definitions for the 
central notions of \emph{allocation}, \emph{deal}, and \emph{utility}, and
it discusses possible restrictions to the class of admissible deals (both 
structural and in terms of acceptability to individual agents). 
Section~\ref{sec:prelim} also introduces the various concepts of a social
preference we are going to consider in this paper. 
Subsequent sections analyse specific instances of the basic negotiation 
framework (characterised, in particular, by different criteria for the acceptability 
of a proposed deal) with respect to specific notions of social welfare.
In the first instance, agents are assumed to be \emph{rational} (and ``myopic'') 
in the sense of never accepting a deal that would result in a negative payoff. 
Section~\ref{sec:withmoney} analyses the first variant of this model of rational negotiation, 
which  allows for monetary side payments to increase the range of acceptable deals. As we shall see, 
this model facilitates negotiation processes that maximise \emph{utilitarian social welfare}. 
If side payments are not possible, we cannot guarantee outcomes with maximal social welfare,
but it is still possible to negotiate \emph{Pareto optimal} allocations. 
This variant of the rational model is studied in Section~\ref{sec:withoutmoney}. 
Both Section~\ref{sec:withmoney} and~\ref{sec:withoutmoney} also investigate how restrictions 
to the range of utility functions agents may use to model their preferences can 
affect such convergence results. 

In the second part of the paper we apply our methodology to agent societies
where the concept of social welfare is given a different kind of interpretation 
than is commonly the case in the multiagent systems literature. 
Firstly, in Section~\ref{sec:egalitarian} we analyse our framework of resource 
allocation by negotiation in the context of \emph{egalitarian agent societies}.
Then Section~\ref{sec:lorenz} discusses a variant of the framework that combines ideas from
both the utilitarian and the egalitarian programme and enables agents to negotiate
\emph{Lorenz optimal} allocations of resources. 
Finally, Section~\ref{sec:further} introduces the idea of using an \emph{elitist} model of social 
welfare for applications where societies of agents are merely a means of enabling at least one
agent to achieve their goal. This section also reviews the concept of \emph{envy-freeness} and 
discusses ways of measuring different degrees of envy. 

Section~\ref{sec:conclusion} summarises our results and concludes with a 
brief discussion of the concept of \emph{welfare engineering}, \ie\
with the idea of choosing tailor-made definitions of social welfare for
different applications and designing agents' behaviour profiles accordingly.

\section{Preliminaries}
\label{sec:prelim}

The basic scenario of \emph{resource allocation by negotiation} studied 
in this paper is that of an artificial society inhabited by a number of agents, each 
of which initially holds a certain number of resources. These agents will typically 
ascribe different values (utilities) to different bundles of resources. They may then 
engage in negotiation and agree on the reallocation of some of the resources, for 
example, in order to improve their respective individual welfare (\ie\
to increase their utility).
Furthermore, we assume that it is in the interest of the system designer that 
these distributed negotiation processes ---somehow--- also result in a positive 
payoff for society as a whole. 

\subsection{Basic Definitions}\label{sec:basicdefs}

An instance of our abstract negotiation framework consists of a finite set of 
(at least two) \emph{agents} $\cal A$ and a finite set of \emph{resources} $\cal R$. 
Resources are indivisible and non-sharable.
An \emph{allocation of resources} is a partitioning of $\cal R$ amongst the agents in $\cal A$.

\begin{definition}[Allocations]
An allocation of resources is a function~$A$ from $\cal A$ to subsets 
of $\cal R$ such that $A(i)\cap A(j)=\emptyset$ for $i\not=j$ and 
$\bigcup_{i\in{\cal A}}A(i)={\cal R}$.
\end{definition}
For example, given an allocation $A$ with $A(i)=\{r_3,r_7\}$,
agent $i$ would own resources $r_3$ and $r_7$. 
Given a particular allocation of resources, agents may agree on a 
(multilateral) \emph{deal} to exchange some of the resources they 
currently hold. In the most general case, 
any numbers of agents and resources could be involved in a single deal.
From an abstract point of view, a deal takes us from one allocation of resources to 
the next. That is, we may characterise a deal as a pair of allocations.

\begin{definition}[Deals]
\label{def:deals}
A deal is a pair $\delta=(A,A')$ where $A$ and $A'$ are allocations 
of resources with $A\not=A'$.
\end{definition}
The set of agents \emph{involved in a deal} $\delta=(A,A')$ is given by
${\cal A}^\delta = \{i\in{\cal A}\;|\,A(i)\not=A'(i)\}$.
The \emph{composition} of two deals is defined as follows:
If $\delta_1=(A,A')$ and $\delta_2=(A',A'')$, then $\delta_1 \circ \delta_2 = (A,A'')$.
If a given deal is the composition of two deals that concern disjoint sets of agents,
then that deal is said to be \emph{independently decomposable}. 

\begin{definition}[Independently decomposable deals]\label{def:id}
A deal $\delta$ is called independently decomposable iff
there exist deals $\delta_1$ and $\delta_2$ such that
$\delta=\delta_1\circ\delta_2$ and
${\cal A}^{\delta_1} \cap {\cal A}^{\delta_2} = \emptyset$.
\end{definition}
Observe that if
$\delta=(A,A')$ is independently decomposable then there exists an intermediate 
allocation $B$ different from both $A$ and $A'$ such that the intersection 
of $\{i\in{\cal A}\,|\,A(i)\not=B(i)\}$ and $\{i\in{\cal A}\,|\,B(i)\not=A'(i)\}$
is empty, \ie\ such that the union of 
$\{i\in{\cal A}\,|\,A(i)=B(i)\}$ and $\{i\in{\cal A}\,|\,B(i)=A'(i)\}$ 
is the full set of agents ${\cal A}$.
Hence, $\delta=(A,A')$ \emph{not} being independently decomposable implies
that there exists no allocation $B$ different from both $A$ and $A'$ 
such that either $B(i)=A(i)$ or $B(i)=A'(i)$ for all agents $i\in{\cal A}$
(we are going to use this fact in the proofs of our ``necessity theorems'' later on).

The value an agent $i\in{\cal A}$ ascribes to a particular set of resources $R$
will be modelled by means of a \emph{utility function}, that is, a function 
from sets of resources to real numbers. We are going to consider both 
general utility functions (without any restrictions) and several more specific
classes of functions.

\begin{definition}[Utility functions]
Every agent $i\in{\cal A}$ is equipped with a utility function 
$u_i:2^{\cal R}\to\R$. We are going to consider the following 
restricted classes of utility functions:
\begin{itemize}
\item $u_i$ is non-negative iff $u_i(R)\geq 0$ for all $R\subseteq{\cal R}$.
\item $u_i$ is positive iff it is non-negative and $u_i(R)\not=0$ for all 
$R\subseteq{\cal R}$ with $R\not=\emptyset$.
\item $u_i$ is monotonic iff $R_1\subseteq R_2$ implies 
$u_i(R_1)\leq u_i(R_2)$ for all $R_1, R_2\subseteq {\cal R}$.
\item $u_i$ is additive iff $u_i(R)=\sum_{r\in R}u_i(\{r\})$ for all $R\subseteq{\cal R}$.
\item $u_i$ is a 0-1 function iff it is additive and
$u_i(\{r\})=0$ or $u_i(\{r\})=1$ for all $r\in{\cal R}$.
\item $u_i$ is dichotomous iff $u_i(R)=0$ or $u_i(R)=1$ for all $R\subseteq{\cal R}$.
\end{itemize}
\end{definition}
Recall that, given an allocation $A$, the set $A(i)$ is 
the bundle of resources held by agent $i$ in that situation.
We are usually going to abbreviate $u_i(A) = u_i(A(i))$ for 
the utility value assigned by agent $i$ to that bundle.

\subsection{Deal Types and Rationality Criteria}
\label{sec:dealclasses}

In this paper we investigate what kinds of negotiation outcomes agents can
achieve by using different \emph{classes of deals}. A class of deals 
may be characterised by both \emph{structural} constraints (number of
agents and resources involved, etc.) and \emph{rationality} constraints
(relating to the changes in utility experienced by the agents involved).

Following \citeA{Sandholm1998a}, we can distinguish a number of structurally 
different types of deals. The most basic are \emph{1-deals}, 
where a single item is passed from one agent to another.

\begin{definition}[1-deals]
A 1-deal is a deal involving the reallocation of exactly one resource.
\end{definition}
This corresponds to the ``classical'' form of a contract typically found 
in the Contract Net protocol \cite{Smith1980a}.
Deals where one agent passes a set of resources on to another agent are called 
\emph{cluster deals}. Deals where one agent gives a single item to another agent 
who returns another single item are called \emph{swap deals}.
Sometimes it can also be necessary to exchange resources between more than just 
two agents. In Sandholm's terminology, a \emph{multiagent deal} is a deal that
could involve any number of agents, where each agent passes at most one resource
to each of the other agents taking part. Finally, deals that combine the features 
of the cluster and the multiagent deal type are called
\emph{combined deals} by Sandholm. These could involve any number of agents 
and any number of resources.
Therefore, \emph{every} deal $\delta$, in the sense of Definition~\ref{def:deals},
is a combined deal. In the remainder of this paper, when speaking about deals without
further specifying their type, we are always going to refer to combined deals (without
any structural restrictions).\footnote{%
The ontology of deal types discussed here is, of course, not exhaustive.
It would, for instance, also be of interest to consider the class of bilateral deals
(involving exactly two agents but any number of items).}

An agent may or may not find a particular deal $\delta$ acceptable.
Whether or not an agent will accept a given deal depends on the 
\emph{rationality criterion} it applies when evaluating deals.
A selfish agent~$i$ may, for instance, only accept deals $\delta=(A,A')$
that strictly improve its personal welfare: $u_i(A)<u_i(A')$.
We call criteria such as this, which only depend on the utilities 
of the agent in question, \emph{personal} rationality criteria.
While we do not want to admit arbitrary rationality criteria, the
classes of deals that can be characterised using personal rationality 
criteria alone is somewhat too narrow for our purposes. 
Instead, we are going to consider rationality criteria that are 
\emph{local} in the sense of only depending on the utility levels
of the agents involved in the deal concerned.

\begin{definition}[Local rationality criteria]\label{def:locality}
A class $\Delta$ of deals is said to be characterised by a local 
rationality criterion iff it is possible to define a predicate $\Phi$ 
over $2^{{\cal A}\times\R\times\R}$ such that a deal $\delta=(A,A')$ belongs to $\Delta$ 
iff $\Phi(\{(i,u_i(A),u_i(A'))\;|\;i\in{\cal A}^\delta\})$ holds true.
\end{definition}
That is, $\Phi$ is mapping a set of triples of one agent name and two reals (utilities) 
each to truth values. The locality aspect comes in by only applying $\Phi$ to 
the set of triples for those agents whose bundle changes with $\delta$.
Therefore, for instance, the class of all deals that increase the utility
of the previously poorest agent is not characterisable by a local rationality
criterion (because this condition can only be checked by inspecting the
utilities of all the agents in the system). 

\subsection{Socially Optimal Allocations of Resources}
\label{sec:swos}

As already mentioned in the introduction,
we may think of a multiagent system as a \emph{society} of autonomous software agents.
While agents make their \emph{local} decisions on what deals to propose and to
accept, we can also analyse the system from a \emph{global} or \emph{societal} 
point of view and may thus prefer certain allocations of resources over others.
To this end, welfare economics provides formal tools to assess how the 
distribution of resources amongst the members of a society affects the well-being 
of society as a whole \cite{Sen1970,Moulin1988a,social-choice-handbook}.

Given the preference profiles of the individual agents in a society (which, 
in our framework, are represented by means of their utility functions), a 
\emph{social welfare ordering} over alternative allocations of resources
formalises the notion of a society's preferences.
Next we are going to formally introduce the most important social welfare orderings
considered in this paper (some additional concepts of social welfare are discussed
towards the end of the paper). In some cases, social welfare is best defined in 
terms of a collective utility function.
One such example is the notion of \emph{utilitarian} social welfare.                         

\begin{definition}[Utilitarian social welfare]
\label{def:swu}
The utilitarian social welfare $sw_u(A)$ of an allocation of resources $A$ 
is defined as follows:
\[\begin{array}{rcl}
sw_u(A) & = & {\displaystyle\sum_{i\in{\cal A}}u_i(A)}
\end{array}\]
\end{definition}
Observe that maximising the collective utility function $sw_u$ amounts to maximising 
the \emph{average utility} enjoyed by the agents in the system.
Asking for maximal utilitarian social welfare is a very strong requirement. 
A somewhat weaker concept is that of \emph{Pareto optimality}.
An allocation is Pareto optimal iff there is no other allocation
with higher utilitarian social welfare that would be no worse for any of
of the agents in the system (\ie\ 
that would be strictly better for at least one agent without being worse for
any of the others).
\begin{definition}[Pareto optimality]\label{def:pareto}
An allocation $A$ is called Pareto optimal iff there is no 
allocation $A'$ such that $sw_u(A)<sw_u(A')$ and $u_i(A)\leq u_i(A')$ for 
all $i\in{\cal A}$. 
\end{definition}
The first goal of an \emph{egalitarian} society should be to increase
the welfare of its weakest member \cite{Rawls1971a,Sen1970}.
In other words, we can measure the social welfare of such a society by
measuring the welfare of the agent that is currently worst off.

\begin{definition}[Egalitarian social welfare]\label{def:swe}
The egalitarian social welfare $sw_e(A)$ of an allocation of resources $A$ is
defined as follows:
\[\begin{array}{rcl}
sw_e(A) & = & \min\{u_i(A)\;|\;i\in{\cal A}\}
\end{array}\]
\end{definition}
The egalitarian collective utility function $sw_e$ gives rise to a social 
welfare ordering over alternative allocations of resources: $A'$ is strictly 
preferred over $A$ iff $sw_e(A)<sw_e(A')$. 
This ordering is sometimes called the \emph{maximin-ordering}.
The maximin-ordering only takes into account the welfare
of the currently weakest agent, but is insensitive to utility fluctuations in the
rest of society. To allow for a finer distinction of the social welfare of different
allocations we introduce the so-called \emph{leximin-ordering}.

For a society with $n$ agents, let $\{u_1,\ldots,u_n\}$ be the set
of utility functions for that society. Then  every allocation $A$
determines a utility vector $\langle u_1(A),\ldots,u_n(A)\rangle$
of length $n$. If we rearrange the elements of that vector in
increasing order we obtain the \emph{ordered utility vector} for
allocation $A$, which we are going to denote by $\vec{u}(A)$.
The number $\vec{u}_i(A)$ is the $i$th element in such a vector
(for $1\leq i\leq |{\cal A}|$). That is, $\vec{u}_1(A)$ for instance, is 
the utility value assigned to allocation $A$ by the currently weakest agent.
We now declare a \emph{lexicographic ordering} over vectors 
of real numbers (such as $\vec{u}(A)$) in the usual
way: $\vec{x}$ lexicographically precedes $\vec{y}$ iff
$\vec{x}$ is a (proper) prefix of $\vec{y}$ or $\vec{x}$
and  $\vec{y}$ share a common (proper) prefix of length $k$
(which may be $0$) and we have $\vec{x}_{k+1}<\vec{y}_{k+1}$.

\begin{definition}[Leximin-ordering]\label{def:leximin}
The leximin-ordering $\prec$ over alternative allocations of resources 
is defined as follows: 
\[\begin{array}{rcl}
A\prec A' & \mbox{ iff }\, & 
\mbox{$\vec{u}(A)$ lexicographically precedes $\vec{u}(A')$}
\end{array}\]
\end{definition}
We write $A\preceq A'$ iff either $A\prec A'$ or $\vec{u}(A)=\vec{u}(A')$. 
An allocation of resources $A$ is called \emph{leximin-maximal} iff there is no
other allocation $A'$ such that $A\prec A'$.


Finally, we introduce the concept of \emph{Lorenz domination},
a social welfare ordering that combines utilitarian
and egalitarian aspects of social welfare. The basic idea is to endorse
deals that result in an improvement with respect to utilitarian welfare
without causing a loss in egalitarian welfare, and vice versa.

\begin{definition}[Lorenz domination]\label{def:lorenz}
Let $A$ and $A'$ be allocations for a society with $n$ agents.
Then $A$ is Lorenz dominated by $A'$ iff 
\[\begin{array}{rcl}
\displaystyle\sum_{i=1}^k \vec{u}_i(A) & \leq & 
\displaystyle\sum_{i=1}^k \vec{u}_i(A')
\end{array}\]
for all $k$ with $1\leq k\leq n$ and, furthermore, that inequality is strict
for at least one $k$.
\end{definition}
For any $k$ with $1\leq k\leq n$, the sum referred to in the above definition
is the sum of the utility values assigned to the respective allocation of
resources by the $k$ weakest agents.
For $k=1$, this sum is equivalent to the egalitarian social welfare for that
allocation. For $k=n$, it is equivalent to the utilitarian social welfare.
An allocation of resources is called \emph{Lorenz optimal} iff it is not
Lorenz dominated by any other allocation. 
 
We illustrate some of the above social welfare concepts (and the use of ordered utility vectors)
by means of an example. Consider a society with three agents and two resources,
with the agents' utility functions given by the following table:

\begin{center}
\begin{tabular}{lll} \hline
$u_1(\emptyset) = 0$   & $u_2(\emptyset) = 0$   & $u_3(\emptyset) = 0$ \\
$u_1(\{r_1\}) = 5$     & $u_2(\{r_1\}) = 4$     & $u_3(\{r_1\}) = 2$ \\ 
$u_1(\{r_2\}) = 3$     & $u_2(\{r_2\}) = 2$     & $u_3(\{r_2\}) = 6$ \\ 
$u_1(\{r_1,r_2\}) = 8$ & $u_2(\{r_1,r_2\}) = 17$ & $u_3(\{r_1,r_2\}) = 7$\\ 
\hline
\end{tabular}
\end{center} 
First of all, we observe that the egalitarian social welfare will be $0$ for any
possible allocation in this scenario, because at least one of the agents
would not get any resources at all. Let $A$ be the allocation where 
agent~2 holds the full bundle of resources. Observe that this is the 
allocation with maximal utilitarian social welfare. The corresponding utility
vector is $\langle 0,17,0\rangle$, \ie\ $\vec{u}(A)=\langle 0,0, 17\rangle$.
Furthermore, let $A'$ be the allocation where agent~1 gets $r_1$, agent~2 gets 
$r_2$, and agent~3 has to be content with the empty bundle. Now we get
an ordered utility vector of $\langle 0,2,5\rangle$.
The initial element in either vector is $0$, but $0<2$, \ie\
$\vec{u}(A)$ lexicographically precedes $\vec{u}(A')$.
Hence, we get $A\prec A'$, \ie\
$A'$ would be the socially preferred allocation 
with respect to the leximin-ordering.
Furthermore, both $A$ and $A'$ are Pareto optimal and neither
is Lorenz-dominated by the other.
Starting from allocation $A'$, agents~1 and~2 swapping their respective bundles would result 
in an allocation with the ordered utility vector $\langle 0,3,4\rangle$, \ie\
this move would result in a Lorenz improvement.

\section{Rational Negotiation with Side Payments}
\label{sec:withmoney}

In this section, we are going to discuss a first instance of the general
framework of resource allocation by negotiation set out earlier. 
This particular variant, which we shall refer to as the model of
\emph{rational negotiation with side payments} (or simply \emph{with money}), 
is equivalent to a framework put forward by Sandholm where agents negotiate in 
order to reallocate \emph{tasks} \cite{Sandholm1998a}. 
For this variant of the framework, our aim will be to negotiate allocations
with \emph{maximal utilitarian social welfare}.

\subsection{Individual Rationality}

In this instance of our negotiation framework,
a deal may be accompanied by a number of monetary \emph{side payments} 
to compensate some of the agents involved for accepting a loss in utility. 
Rather than specifying for each pair of agents how much money the former is
supposed to pay to the latter, we simply say how much money each agent either 
pays out or receives. This can be modelled by using what we call a 
\emph{payment function}.

\begin{definition}[Payment functions]
A payment function is a function $p$ from $\cal A$ 
to real numbers satisfying the following condition:
\[\begin{array}{rcl}
\displaystyle\sum_{i\in{\cal A}}p(i) & = & 0
\end{array}\] 
\end{definition}
Here, $p(i)>0$ means that agent~$i$ \emph{pays} the amount of $p(i)$, while 
$p(i)<0$ means that it \emph{receives} the amount of $-p(i)$. By definition
of a payment function, the sum of all payments is $0$, \ie\
the overall amount of money present in the system does not change.\footnote{%
As the overall amount of money present in the system stays constant throughout the 
negotiation process, it makes sense not to take it into account for the evaluation 
of social welfare.}

In the rational negotiation model, agents are self-interested in the sense of
only proposing or accepting deals that strictly increase their own welfare
\cite<for a justification of this approach we refer to>{Sandholm1998a}. 
This ``myopic'' notion of \emph{individual rationality} may be formalised as follows.

\begin{definition}[Individual rationality]
\label{def:individualrationality}
A deal $\delta=(A,A')$ is called individually rational iff there exists a payment 
function $p$ such that $u_i(A')-u_i(A)>p(i)$ for all $i\in{\cal A}$, 
except possibly $p(i)=0$ for agents $i$ with $A(i)=A'(i)$.
\end{definition}
That is, agent~$i$ will be prepared to accept the deal $\delta$ iff it has to pay 
less than its gain in utility or it will get paid more than its loss in utility,
respectively. Only for agents $i$ not affected by the deal, \ie\
in case  $A(i)=A'(i)$, there may be no payment at all.
For example, if $u_i(A)=8$ and $u_i(A')=5$, then the utility of agent~$i$
would be reduced by $3$ units if it were to accept the deal $\delta=(A,A')$.
Agent~$i$ will only agree to this deal if it is accompanied by a side payment 
of more than $3$ units; that is, if the payment function $p$ satisfies
$-3>p(i)$.

For any given deal, there will usually be a range of possible side payments. 
How agents manage to agree on a particular one is not a matter of 
consideration at the abstract level at which we are discussing this 
framework here. We assume that a deal will go ahead as long as there 
exists \emph{some} suitable payment function $p$.
We should point out that this assumption may not be justified under
all circumstances. For instance, if utility functions are not publicly 
known and agents are risk-takers, then a potential deal may not be
identified as such, because some of the agents may understate their  
interest in that deal in order to maximise their expected 
payoff \cite{MyersonSatterthwaite1983}. Therefore, the theoretical results
on the reachability of socially optimal allocations 
reported below will only apply under the assumption that such strategic
considerations will not prevent agents from making mutually beneficial
deals.

\subsection{An Example}\label{sec:examplewithmoney}

As an example, consider a system with two agents, 
agent~1 and agent~2, and a set of two resources ${\cal R}=\{r_1,r_2\}$. The following table
specifies the values of the utility functions $u_1$ and $u_2$ for every subset 
of $\{r_1,r_2\}$:

\begin{center}
\begin{tabular}{ll} \hline
$u_1(\emptyset) = 0$ & $u_2(\emptyset) = 0$ \\
$u_1(\{r_1\}) = 2$ & $u_2(\{r_1\}) = 3$ \\ 
$u_1(\{r_2\}) = 3$ & $u_2(\{r_2\}) = 3$ \\ 
$u_1(\{r_1,r_2\}) = 7$ & $u_2(\{r_1,r_2\}) = 8$ \\ \hline 
\end{tabular}
\end{center}
Also suppose agent~1 initially holds the full set of resources $ \{r_1,r_2\}$
and agent~2 does not own any resources to begin with.

The utilitarian social welfare for this initial allocation is $7$, but it could be $8$, 
namely if agent~2 had both resources. As we are going to see next, the simple class of
\emph{1-deals} alone are not always sufficient to 
guarantee the optimal outcome of a negotiation process (if agents abide to the 
individual rationality criterion for the acceptability of a deal). 
In our example, the only possible 1-deals
would be to pass either $r_1$ or $r_2$ from agent~1 to agent~2. In either case, the
loss in utility incurred by agent~1 ($5$ or $4$, respectively) would outweigh the gain 
of agent~2 ($3$ for either deal), so there is no payment function that would make these
deals individually rational. 
The \emph{cluster deal} of passing $\{r_1,r_2\}$ from agent~1 to~2, on the other hand,
\emph{would} be individually rational if agent~2 paid agent~1 an amount of, say, $7.5$
units.

Similarly to the example above, we can also 
construct scenarios where swap deals or multiagent deals are necessary (\ie\
where cluster deals alone would not be sufficient to guarantee maximal social welfare).
This also follows from Theorem~\ref{thm:necessarywithmoney}, which we are going to present
later on in this section. Several concrete examples are given by \citeA{Sandholm1998a}.

\subsection{Linking Individual Rationality and Social Welfare}

The following result, first stated in this form by \citeA{EndrissEtAlAAMAS2003-optimal}, 
says that a deal (with money) is individually rational iff it increases 
utilitarian social welfare. We are mainly going to use this lemma 
to give a simple proof of Sandholm's main result on sufficient contract 
types \cite{Sandholm1998a}, but it has also found useful applications in its 
own right \cite{DunneEtAlAIJ,DunneJAIR2005,EndrissMaudetJAAMAS2005}.

\begin{lemma}[Individually rational deals and utilitarian social welfare]
\label{lem:rationaldeals}
A deal $\delta=(A,A')$ is individually rational iff $sw_u(A)<sw_u(A')$.
\end{lemma} 

\begin{proof}
\lefttorightproof\
By definition, $\delta=(A,A')$ is individually rational iff there exists 
a payment function $p$ such that $u_i(A')-u_i(A)>p(i)$ holds for all 
$i\in{\cal A}$, except possibly $p(i)=0$ in case $A(i)=A'(i)$. 
If we add up the inequations for all agents $i\in{\cal A}$ we get:
\[\begin{array}{rcl}
{\displaystyle \sum_{i\in{\cal A}}(u_i(A')-u_i(A))} & > & 
{\displaystyle \sum_{i\in{\cal A}}p(i)}
\end{array}\] 
By definition of a payment function, the righthand side equates to $0$ while,
by definition of utilitarian social welfare, the lefthand side equals 
$sw_u(A')-sw_u(A)$. Hence, we really get $sw_u(A)<sw_u(A')$ as claimed.

\righttoleftproof\
Now let $sw_u(A)<sw_u(A')$. We have to show that $\delta=(A,A')$ is an individually 
rational deal. We are done if we can prove that there exists a payment 
function $p$ such that $u_i(A')-u_i(A)>p(i)$ for all $i\in{\cal A}$.
We define the function $p:{\cal A}\to\R$ as follows:
\[\begin{array}{rcl}
p(i) & = & {\displaystyle u_i(A')-u_i(A)-\frac{sw_u(A')-sw_u(A)}{|{\cal A}|}}
\quad(\mbox{for}\ i\in{\cal A})
\end{array}\]
First, observe that $p$ really \emph{is} a payment function, 
because we get $\sum_{i\in{\cal A}}p(i)=0$. We also get  
$u_i(A')-u_i(A)>p(i)$ for all $i\in{\cal A}$, because we
have $sw_u(A')-sw_u(A')>0$. Hence, $\delta$ must indeed be an 
individually rational deal. 
\end{proof}

\afterproof
Lemma~\ref{lem:rationaldeals} suggests that the 
function $sw_u$ does indeed provide an appropriate measure of social well-being in societies of 
agents that use the notion of individual rationality (as given by Definition~\ref{def:individualrationality}) 
to guide their behaviour during negotiation.
It also shows that individual rationality is indeed a local rationality criterion
in the sense of Definition~\ref{def:locality}.

\subsection{Maximising Utilitarian Social Welfare}

Our next aim is to show that any sequence of deals in the rational negotiation model 
with side payments will converge to an allocation with maximal utilitarian 
social welfare; that is, the class of individually rational deals (as given by
Definition~\ref{def:individualrationality}) is sufficient to guarantee
optimal outcomes for agent societies measuring welfare according to the utilitarian
programme (Definition~\ref{def:swu}). This has originally been shown by \citeA{Sandholm1998a} 
in the context of a framework where rational agents negotiate 
in order to reallocate tasks and where the global aim is to minimise the overall 
costs of carrying out these tasks. 

\begin{theorem}[Maximal utilitarian social welfare]
\label{thm:sufficientwithmoney}
Any sequence of individually rational deals will eventually result in an 
allocation with maximal utilitarian social welfare.
\end{theorem}

\begin{proof}
Given that both the set of agents $\cal A$ as well as the set of resources
$\cal R$ 
are required to be finite, there can be only
a finite number of distinct allocations of resources.
Furthermore, by Lemma~\ref{lem:rationaldeals}, any individually 
rational deal will strictly increase utilitarian social welfare. 
Hence, negotiation must terminate after a finite number of deals.
For the sake of contradiction, assume that the terminal 
allocation $A$ does \emph{not} have maximal utilitarian social welfare, \ie\
there exists another allocation $A'$ with $sw_u(A)<sw_u(A')$.
But then, by Lemma~\ref{lem:rationaldeals}, the deal $\delta=(A,A')$ 
would be individually rational and thereby possible, which
contradicts our earlier assumption of $A$ being a terminal allocation.
\end{proof}

\afterproof
At first sight, this result may seem almost trivial. The notion of a multilateral
deal without any structural restrictions is a \emph{very} powerful one. A single 
such deal allows for any number of resources to be moved between any number of agents. 
From this point of view, it is not particularly surprising that we can always reach 
an optimal allocation (even in just a single step!).
Furthermore, \emph{finding} a suitable deal is a very complex task, which may not 
always be viable in practice. 
The crucial point of Theorem~\ref{thm:sufficientwithmoney} is that
\emph{any} sequence of deals will result in an optimal allocation.
That is, whatever deals are agreed on in the early stages of negotiation,
the system will never get stuck in a local optimum and finding an allocation
with maximal social welfare remains an option throughout (provided, of course,
that agents are actually able to identify any deal that is theoretically possible).
Given the restriction to deals that are individually rational for all the
agents involved, social welfare must increase with every single deal.
Therefore, negotiation always pays off, even if it has to stop early due
to computational limitations.
 
The issue of complexity is still an important one.
If the full range of deals is too large to be managed in practice,
it is important to investigate how close we can get to finding an optimal
allocation if we restrict the set of allowed deals to certain simple patterns.
\citeA{AnderssonSandholmICDCS2000}, for instance, have conducted
a number of experiments on the sequencing of certain contract/deal types to reach
the best possible allocations within a limited amount of time.
For a complexity-theoretic analysis of the problem of deciding whether 
it is possible to reach an optimal allocation by means of structurally
simple types of deals (in particular 1-deals), we
refer to recent work by \citeA{DunneEtAlAIJ}.

\subsection{Necessary Deals}

The next theorem improves upon Sandholm's main result
regarding necessary contract types \cite{Sandholm1998a},
by extending it to the cases where either all utility functions are monotonic 
or all utility functions are dichotomous.\footnote{%
In fact, our theorem not only sharpens but also also corrects a mistake in previous expositions 
of this result \cite{Sandholm1998a,EndrissEtAlAAMAS2003-optimal}, where the restriction 
to deals that are not independently decomposable had been omitted.}
Sandholm's original result, translated into our terminology, states that for any system 
(consisting of a set of agents ${\cal A}$ and a set of resources ${\cal R}$) 
and any (not independently decomposable) deal $\delta$ for that system, 
it is possible to construct utility functions and choose an initial allocation 
of resources such that $\delta$ is \emph{necessary} to reach an optimal
allocation, if agents only agree to individually rational deals.
All other findings on the insufficiency of certain types of contracts
reported by \citeA{Sandholm1998a} may be considered corollaries to this.
For instance, the fact that, say, cluster deals alone are not sufficient
to guarantee optimal outcomes follows from this theorem if we take $\delta$
to be any particular swap deal for the system in question.

\begin{theorem}[Necessary deals with side payments]
\label{thm:necessarywithmoney}
Let the sets of agents and resources be fixed.
Then for every deal $\delta$ that is not independently decomposable, there exist utility functions 
and an initial allocation such that any sequence of individually rational deals leading
to an allocation with maximal utilitarian social welfare must
include $\delta$. This continues to be the case even when either all utility functions are 
required to be monotonic or all utility functions are required to be dichotomous.
\end{theorem}
 
\begin{proof}
Given a set of agents $\cal A$ and a set of resources $\cal R$,
let $\delta=(A,A')$ with $A\not=A'$ be any deal for this system.
We need to show that there are a collection of utility functions and
an initial allocation such that $\delta$ is necessary to reach an allocation
with maximal social welfare. This would be the case if $A'$ had maximal social
welfare, $A$ had the second highest social welfare, and $A$ were the initial
allocation of resources. 

We first prove the existence of such a collection of functions for
the case where all utility functions are required to be monotonic.
Fix $\epsilon$ such that $0<\epsilon<1$.
As we have $A\not=A'$, there must be an agent $j\in{\cal A}$ such that $A(j)\not=A'(j)$. 
We now define utility functions $u_i$ for agents $i\in{\cal A}$ and sets of resources 
$R\subseteq{\cal R}$ as follows:
\[\begin{array}{rcl}
u_i(R) & = & \left\{\begin{array}{ll}
|R|+\epsilon & \mbox{if}\ R=A'(i)\ \mbox{or}\ (R=A(i)\ \mbox{and}\ i\not=j) \\
|R|          & \mbox{otherwise}
\end{array}\right.
\end{array}\] 
Observe that $u_i$ is a monotonic utility function for every $i\in{\cal A}$.
We get $sw_u(A')=|{\cal R}| + \epsilon\cdot|{\cal A}|$ and $sw_u(A)=sw_u(A')-\epsilon$.
Because $\delta=(A,A')$ is not individually decomposable, there exists
no allocation $B$ different from both $A$ and $A'$ such that either
$B(i)=A(i)$ or $B(i)=A'(i)$ for all agents $i\in{\cal A}$. Hence,
$sw_u(B)\leq sw_u(A)$ for any other allocation $B$.
That is, $A'$ is the (unique) allocation with maximal social welfare and the only allocation with
higher social welfare than $A$. Therefore, if we were to make $A$ the initial
allocation then $\delta=(A,A')$ would be the only deal increasing social
welfare. By Lemma~\ref{lem:rationaldeals}, this means that
$\delta$ would be the only individually rational (and thereby the only possible)
deal. Hence, $\delta$ is indeed necessary to achieve maximal utilitarian 
social welfare.

The proof for the case of dichotomous utility functions is very similar;
we only need to show that a suitable collection of dichotomous utility
functions can be constructed. Again, let $j$ be an agent with $A(j)\not=A'(j)$.
We can use the following collection of functions:
\[\begin{array}{rcl}
u_i(R) & = & \left\{\begin{array}{ll}
1 & \mbox{if}\ R=A'(i)\ \mbox{or}\ (R=A(i)\ \mbox{and}\ i\not=j) \\
0 & \mbox{otherwise}
\end{array}\right.
\end{array}\]
We get $sw_u(A')=|{\cal A}|$, $sw_u(A)=sw_u(A')-1$ and $sw_u(B)\leq sw_u(A)$
for all other allocations $B$. Hence, $A$ is not socially optimal and, with $A$
as the initial allocation, $\delta$ would be the only deal that is individually 
rational.
\end{proof}

\afterproof
We should stress that the set of deals that are not independently decomposable
includes deals involving any number of agents and/or any number of resources.
Hence, by Theorem~\ref{thm:necessarywithmoney}, any negotiation protocol that 
puts restrictions on the structural complexity of deals that may be proposed 
will fail to guarantee optimal outcomes, even when there are no constraints
on either time or computational resources. 
This emphasises the high complexity of our negotiation framework
\cite<see also>{DunneEtAlAIJ,ChevaleyreEtAlCSDT2004,EndrissMaudetJAAMAS2005,DunneJAIR2005}.
The fact that the necessity of (almost) the full range of deals persists, even when all
utility functions are subject to certain restrictions makes this result even more striking.
This is true in particular for the case of dichotomous functions, which are in some sense
the simplest class of utility functions (as they can only distinguish between 
``good'' and ``bad'' bundles).

To see that the restriction to deals that are not independently decomposable
matters, consider a scenario with four agents and two resources. If the deal $\delta$
of moving $r_1$ from agent $1$ to agent $2$, and $r_2$ from agent $3$
to agent $4$ is individually rational, then so will be either one of
the two ``subdeals'' of moving either $r_1$ from agent $1$ to agent $2$ or 
$r_2$ from agent $3$ to agent $4$. Hence, the deal $\delta$ (which is
independently decomposable) cannot be necessary in the sense of 
Theorem~\ref{thm:necessarywithmoney} (with reference to our proof above,
in the case of $\delta$ there \emph{are} allocations $B$ such that either
$B(i)=A(i)$ or $B(i)=A'(i)$ for all agents $i\in{\cal A}$, \ie\ we could get
$sw_u(B)=sw_u(A')$).

\subsection{Additive Scenarios}

Theorem~\ref{thm:necessarywithmoney} is a negative result, because it shows that deals 
of any complexity may be required to guarantee optimal outcomes of negotiation. 
This is partly a consequence of the high degree of generality of our framework.
In Section~\ref{sec:basicdefs}, we have defined utility functions as \emph{arbitrary}
functions from sets of resources to real numbers. For many application domains this may be
unnecessarily general or even inappropriate and we may be able to obtain stronger results
for specific classes of utility functions that are subject to certain restrictions. 
Of course, we have already seen that this is \emph{not} the case for either monotonicity
(possibly the most natural restriction) or dichotomy (possibly the most severe
restriction).

Here, we are going to consider the case of \emph{additive} utility functions, which
are appropriate for domains where combining resources does not result in any synergy 
effects (in the sense of increasing an agent's welfare). We refer to systems where all 
agents have additive utility functions as \emph{additive scenarios}. The following theorem 
shows that for these additive scenarios 1-deals are sufficient to guarantee outcomes
with maximal social welfare.\footnote{This has also been observed by T.~Sandholm
(personal communication, September~2002).}

\begin{theorem}[Additive scenarios]
\label{thm:additive}
In additive scenarios, any sequence of individually rational 
1-deals will eventually result in an allocation 
with maximal utilitarian social welfare.
\end{theorem}

\begin{proof}
Termination is shown as for Theorem~\ref{thm:sufficientwithmoney}.
We are going to show that, whenever the current allocation does not have
maximal social welfare, then there is still a possible 1-deal that is individually rational.

In additive domains, the utilitarian social welfare of a given allocation may be computed
by adding up the appropriate utility values for all the single resources in $\cal R$. 
For any allocation $A$,
let $f_A$ be the function mapping each resource $r\in{\cal R}$ to the agent
$i\in{\cal A}$ that holds $r$ in situation $A$ (that is, we have 
$f_A(r)=i$ iff $r\in A(i)$).
The utilitarian social welfare for allocation $A$ is then given by the 
following formula:
\[\begin{array}{rcl}
sw_u(A) & = & {\displaystyle \sum_{r\in{\cal R}}u_{f_A(r)}(\{r\})}
\end{array}\]
Now suppose that negotiation has terminated with allocation $A$ and
there are no more individually rational 1-deals possible.
Furthermore, for the sake of contradiction, assume that
$A$ is \emph{not} an allocation with maximal social welfare, \ie\
there exists another allocation $A'$ with $sw_u(A)<sw_u(A')$.
But then, by the above characterisation of social welfare for additive
scenarios, there must be at least one resource $r\in{\cal R}$
such that $u_{f_A(r)}(\{r\})  <  u_{f_{A'}(r)}(\{r\})$.
That is, the 1-deal $\delta$ of passing $r$ from agent 
$f_A(r)$ on to agent $f_{A'}(r)$ would increase social welfare. Therefore, 
by Lemma~\ref{lem:rationaldeals}, $\delta$ must be an individually rational deal, \ie\
contrary to our earlier assumption, $A$ cannot be a terminal allocation.
Hence, $A$ must be an allocation with maximal utilitarian social welfare.
\end{proof}

\afterproof
We conclude this section by briefly mentioning
two recent results that both extend, in different ways, the result
stated in Theorem~\ref{thm:additive} (a detailed discussion, however, would be 
beyond the scope of the present paper).
The first of these results shows that rational deals involving at most $k$ resources each
are sufficient for convergence to an allocation with maximal social welfare whenever 
all utility functions are \emph{additively separable} with respect to a common
partition of $\cal R$ ---\ie\ synergies across different parts of the partition are not 
possible and overall utility is defined as the sum of utilities for the different sets 
in the partition \cite{Fishburn1970}--- and each set in this partition has at most $k$ 
elements \cite{ChevaleyreEtAlAAMAS2005}.

The second result concerns a \emph{maximality property} of utility functions
with respect to 1-deals. 
\citeA{ChevaleyreEtAlIJCAI2005} show that the class of \emph{modular} utility
functions, which is only slightly more general than the class of additive functions
considered here (namely, it is possible to assign a non-zero utility to the empty 
bundle), is maximal in the sense that for no class of functions strictly including the
class of modular functions it would still be possible to guarantee that agents using 
utility functions from that larger class and negotiating only individually rational 
1-deals will eventually reach an allocation with maximal utilitarian social welfare 
in all cases.

\section{Rational Negotiation without Side Payments}
\label{sec:withoutmoney}

An implicit assumption made in the framework that we have presented so far
is that every agent has got an \emph{unlimited amount of money} available to it 
to be able to pay other agents whenever this is required for a deal that would
increase utilitarian social welfare. Concretely, if $A$ is the initial allocation and $A'$ is
the allocation with maximal utilitarian social welfare, then agent $i$ may require an amount of
money just below the difference $u_i(A')-u_i(A)$ to be able to get through the negotiation
process. In the context of task contracting, for which this framework has been proposed 
originally \cite{Sandholm1998a}, this may be justifiable, but for resource allocation problems
it seems questionable to make assumptions about the unlimited availability of one particular 
resource, namely money.
In this section, we therefore investigate to what extent
the theoretical results discussed in the previous section
persist to apply when we consider negotiation processes \emph{without} monetary
side payments.

\subsection{An Example}

In a scenario without money, that is, if we do not allow for compensatory payments, 
we cannot always guarantee an outcome with maximal utilitarian social welfare. 
To see this, consider the following simple problem for a system with two agents, 
agent~1 and agent~2, and a single resource $r$. The agents' utility functions are 
defined as follows:

\begin{center}
\begin{tabular}{ll} \hline
$u_1(\emptyset) = 0$ & $u_2(\emptyset) = 0$ \\
$u_1(\{r\}) = 4$ & $u_2(\{r\}) = 7$ \\ \hline
\end{tabular}
\end{center}  
Now suppose agent~1 initially owns the resource.
Then passing $r$ from agent~1 to agent~2 would increase utilitarian social welfare by an
amount of $3$. For the framework \emph{with} money, agent~2 could pay agent~1, say,
the amount of $5.5$ units and the deal would be individually rational for both of 
them. Without money (\ie\
if $p\equiv 0$), however, no individually rational deal is possible 
and negotiation must terminate with a non-optimal allocation.

As maximising social welfare is not generally possible, instead we are going 
to investigate whether a \emph{Pareto optimal} outcome (see Definition~\ref{def:pareto})
is possible in the framework without money, and what types of deals are sufficient to 
guarantee this. 

\subsection{Cooperative Rationality}

As will become clear in due course, in order to get the desired convergence result, 
we need to relax the notion of individual rationality a little. For the framework
without money, we also want agents to agree to a deal, if this at least 
maintains their utility (that is, no strict increase is necessary). 
However, we are still going to require at least one agent to strictly increase 
their utility. This could, for instance, be the agent proposing the deal in
question. 
We call deals conforming to this criterion \emph{cooperatively rational}.\footnote{%
Note that ``cooperatively rational'' agents are still essentially rational.
Their willingness to cooperate only extends to cases where they can benefit others
without any loss in utility for themselves.}

\begin{definition}[Cooperative rationality]
\label{def:cooperativerationality}
A deal $\delta=(A,A')$ is called cooperatively rational
iff $u_i(A)\leq u_i(A')$ for all $i\in{\cal A}$ and there is an 
agent $j\in{\cal A}$ such that $u_j(A)<u_j(A')$.
\end{definition}
In analogy to Lemma~\ref{lem:rationaldeals}, we 
still have $sw_u(A)<sw_u(A')$ for any deal $\delta=(A,A')$ that 
is cooperatively rational, but \emph{not} vice versa.
We call the instance of our negotiation framework where all
deals are cooperatively rational (and hence do not include a
monetary component) the model of \emph{rational negotiation 
without side payments}.

\subsection{Ensuring Pareto Optimal Outcomes}

As the next theorem will show, the class of cooperatively rational deals
is sufficient to guarantee a Pareto optimal outcome of money-free negotiation.
It constitutes the analogue to Theorem~\ref{thm:sufficientwithmoney}
for the model of rational negotiation without side payments.

\begin{theorem}[Pareto optimal outcomes]
\label{thm:sufficientwithoutmoney}
Any sequence of cooperatively rational deals will eventually result 
in a Pareto optimal allocation of resources.
\end{theorem}
 
\begin{proof}
Every cooperatively rational deal strictly increases utilitarian social
welfare (this is where we need the condition that at least one agent behaves 
\emph{truly} individually rational for each deal).
Together with the fact that there are only finitely many different
allocations of resources, this implies that any negotiation process will 
eventually terminate.
For the sake of contradiction, assume negotiation ends with allocation $A$,
but $A$ is not Pareto optimal. 
The latter means that there exists another
allocation $A'$ with $sw_u(A)<sw_u(A')$ and $u_i(A)\leq u_i(A')$ for all $i\in{\cal A}$.
If we had  $u_i(A)=u_i(A')$ for all $i\in{\cal A}$, then also  $sw_u(A)=sw_u(A')$; that
is, there must be at least one $j\in{\cal A}$ with $u_j(A)<u_j(A')$. But then
the deal $\delta=(A,A')$ would be cooperatively rational, which contradicts
our assumption of $A$ being a terminal allocation.
\end{proof}

\afterproof
Observe that the proof would not have gone through if deals were required to be strictly
rational (without side payments), as this would necessitate $u_i(A)<u_i(A')$ for all $i\in{\cal A}$.
Cooperative rationality means, for instance, that agents would be prepared to
give away resources to which they assign a utility value of $0$, without expecting anything
in return. In the framework with money, another agent could always offer such an agent an
infinitesimally small amount of money, who would then accept the deal.

Therefore, our proposed weakened notion of rationality seems indeed a very reasonable price
to pay for giving up money.

\subsection{Necessity Result}

As our next result shows, also for the framework without side payments,
deals of any structural complexity may be necessary in order to be able to guarantee an optimal 
outcome of a negotiation.\footnote{This theorem corrects a mistake in the original statement of
the result \cite{EndrissEtAlAAMAS2003-optimal}, where the restriction to deals that are not
independently decomposable had been omitted.} 
Theorem~\ref{thm:necessarywithoutmoney} improves upon previous 
results \cite{EndrissEtAlAAMAS2003-optimal} by showing that this necessity property persists 
also when either all utility functions belong to the class of monotonic functions or all 
utility functions belong to the class of dichotomous functions.
 
\begin{theorem}[Necessary deals without side payments]
\label{thm:necessarywithoutmoney}
Let the sets of agents and resources be fixed.
Then for every deal $\delta$ that is not independently decomposable, there exist utility functions 
and an initial allocation such that any sequence of cooperatively rational deals leading
to a Pareto optimal allocation would have to include $\delta$. 
This continues to be the case even when either all utility functions are 
required to be monotonic or all utility functions are required to be dichotomous.
\end{theorem}

\begin{proof}
The details of this proof are omitted as it is possible to simply reuse the 
construction used in the proof of Theorem~\ref{thm:necessarywithmoney}. 
Observe that the utility functions defined there also guarantee 
$u_i(A)\leq u_i(A')$ for all $i\in{\cal A}$, \ie\
$A$ is not Pareto optimal, but $A'$ is. If we were to make $A$ 
the initial allocation, then $\delta=(A,A')$ would be the only
cooperatively rational deal (as every other deal would decrease social welfare).
\end{proof}

\subsection{0-1 Scenarios}

We conclude our study of the rational negotiation framework without side payments by 
identifying a class of utility functions where we are able to achieve a reduction in 
structural complexity.\footnote{%
As dichotomous functions are a special case of the \emph{non-negative} functions, the full range
of (not independently decomposable) deals is also necessary in scenarios with non-negative functions.
Interestingly, this changes when we restrict ourselves to \emph{positive} utility functions.
Now the result of Theorem~\ref{thm:necessarywithoutmoney} would \emph{not} hold anymore,
because any deal that would involve
a particular agent (with a positive utility function) giving away all its resources
without receiving anything in return could never be cooperatively rational.
Hence, by Theorem~\ref{thm:sufficientwithoutmoney}, such a deal could never 
be necessary to achieve a Pareto optimal allocation either.}
Consider a scenario where 
agents use additive utility functions that assign either $0$ or $1$
to every single resource (this is what we call 0-1 functions).\footnote{%
Recall the distinction between 0-1 functions and dichotomous functions.
The latter assign either 0 or 1 to each \emph{bundle}, while the former
assign either 0 or 1 to each \emph{individual resource} (the utilities of 
bundles then follow from the fact that 0-1 functions are additive).} 
This may be appropriate when we simply wish to distinguish
whether or not the agent \emph{needs} a particular resource (to execute a given plan,
for example). This is, for instance, the case for some of the agents defined
in the work of \citeA{SadriToniTorroniATAL2001}.
As the following theorem shows, for \emph{0-1 scenarios} (\ie\
for systems where all utility functions are 0-1 functions), 1-deals
are sufficient to guarantee convergence to an allocation with maximal 
utilitarian social welfare, even in the framework \emph{without} monetary 
side payments (where all deals are required to be cooperatively rational). 
 
\begin{theorem}[0-1 scenarios]\label{thm:zeroone}
In 0-1 scenarios, any sequence of cooperatively rational 1-deals will eventually 
result in an allocation with maximal utilitarian social welfare.
\end{theorem}
 
\begin{proof}
Termination is shown as in the proof of Theorem~\ref{thm:sufficientwithoutmoney}.
If an allocation $A$ does not have maximal social welfare then it must
be the case that some agent $i$ holds a resource $r$ with $u_i(\{r\})=0$ and
there is another agent $j$ in the system with $u_j(\{r\})=1$. Passing $r$
from $i$ to $j$ would be a cooperatively rational deal, so either
negotiation has not yet terminated or we are indeed in a situation with 
maximal utilitarian social welfare.
\end{proof}

\afterproof 
This result may be interpreted as a formal justification for some of the
negotiation strategies proposed by \citeA{SadriToniTorroniATAL2001}.

\section{Egalitarian Agent Societies}
\label{sec:egalitarian}

In this section, we are going to apply the same methodology that we have used 
to study optimal outcomes of negotiation in systems designed according to
utilitarian principles in the first part of this paper to the analysis of
\emph{egalitarian agent societies}. The classical counterpart to the utilitarian
collective utility function $sw_u$ is the egalitarian collective utility function
$sw_e$ introduced in Definition~\ref{def:swe} \cite{Moulin1988a,Sen1970,Rawls1971a}.
Therefore, we are going to study the design of agent societies for which 
negotiation can be guaranteed to converge to an allocation of resources
with maximal egalitarian social welfare.

Our first aim will be to identify a suitable criterion that 
agents inhabiting an egalitarian agent society may use 
to decide whether or not to accept a particular deal. Clearly, cooperatively
rational deals, for instance, would not be an ideal choice, because 
Pareto optimal allocations will typically not be optimal from an egalitarian 
point of view \cite{Moulin1988a}.

\subsection{Pigou-Dalton Transfers and Equitable Deals}

When searching the economics literature for a class of deals that would
benefit society in an egalitarian system we soon encounter
\emph{Pigou-Dalton transfers}. The Pigou-Dalton \emph{principle} states 
that whenever a utility transfer between two agents takes place which reduces 
the difference in utility between the two, then that transfer should be 
considered socially beneficial \cite{Moulin1988a}. In the context of our 
framework, a Pigou-Dalton transfer (between agents~$i$ and~$j$) can be defined 
as follows.
\begin{definition}[Pigou-Dalton transfers]\label{def:pigoudalton}
A deal $\delta=(A,A')$ is called a Pigou-Dalton transfer iff it satisfies
the following criteria:
\begin{itemize}
\item Only two agents $i$ and $j$ are involved in the deal:
${\cal A}^\delta=\{i,j\}$.
\item The  deal is mean-preserving:
$u_i(A)+u_j(A) = u_i(A')+u_j(A')$.
\item The deal reduces inequality:
$|u_i(A')-u_j(A')| < |u_i(A)-u_j(A)|$.
\end{itemize}
\end{definition}
The second condition in this definition could be relaxed to 
postulate $u_i(A)+u_j(A) \leq u_i(A')+u_j(A')$,
to also allow for inequality-reducing deals that increase overall utility.

Pigou-Dalton transfers capture certain egalitarian principles; but are they
sufficient as acceptability criteria to guarantee negotiation outcomes 
with maximal egalitarian social welfare? Consider the following example:

\begin{center}
\begin{tabular}{ll} \hline
$u_1(\emptyset) = 0$ & $u_2(\emptyset) = 0$ \\
$u_1(\{r_1\}) = 3$ & $u_2(\{r_1\}) = 5$ \\ 
$u_1(\{r_2\}) = 12$ & $u_2(\{r_2\}) = 7$ \\ 
$u_1(\{r_1,r_2\}) = 15$ & $u_2(\{r_1,r_2\}) = 17$ \\ \hline 
\end{tabular}
\end{center}
The first agent attributes a relatively low utility value to 
$r_1$ and a high one to $r_2$. Furthermore, the value of both 
resources together is simply the sum of the individual utilities, \ie\
agent~1 is using an additive utility function (no synergy effects). 
The second agent ascribes a medium value to either resource and a 
very high value to the full set.
Now suppose the initial allocation of resources is $A$
with $A(1)=\{r_1\}$ and $A(2)=\{r_2\}$. The ``inequality index'' 
for this allocation is $|u_1(A)-u_2(A)|=4$.
We can easily check that inequality is in fact minimal for allocation $A$ (which
means that there can be no inequality-reducing deal, and certainly no Pigou-Dalton 
transfer, given this allocation).
However, allocation $A'$ with $A'(1)=\{r_2\}$ and $A'(2)=\{r_1\}$ would 
result in a higher level of egalitarian social welfare (namely $5$ instead of $3$).
Hence, Pigou-Dalton transfers alone are not sufficient to guarantee
optimal outcomes of negotiations in egalitarian agent societies.
We need a more general acceptability criterion.

Intuitively, agents operating according to egalitarian principles should
help any of their fellow agents that are worse off than they are themselves
(as long as they can afford to do so without themselves ending up even worse).
This means, the purpose of any exchange of resources should be to improve
the welfare of the weakest agent involved in the respective deal.
We formalise this idea by introducing the class of \emph{equitable} deals.

\begin{definition}[Equitable deals]\label{def:equitable}
A deal $\delta=(A,A')$ is called equitable iff 
it satisfies the following criterion:
\[\begin{array}{rcl}
\min\{u_i(A)\;|\;i\in{\cal A}^\delta\} & < & \min\{u_i(A')\;|\;i\in{\cal A}^\delta\}
\end{array}\]
\end{definition}
Recall that ${\cal A}^\delta = \{i\in{\cal A}\;|\,A(i)\not=A'(i)\}$
denotes the set of agents involved in the
deal $\delta$. Given that for $\delta=(A,A')$ to be a deal we require 
$A\not=A'$, ${\cal A}^\delta$ can never be the empty set (\ie\
the minima referred to in above definition are well-defined).
Note that equitability is a local rationality criterion in the sense
of Definition~\ref{def:locality}.

It is easy to see that any Pigou-Dalton transfer will also be an equitable 
deal, because it will always result in an improvement for the weaker one of 
the two agents concerned.
The converse, however, does not hold (not even if we restrict ourselves to 
deals involving only two agents). In fact, equitable deals may even increase 
the inequality of the agents concerned, namely in cases where the happier 
agent gains more utility than the weaker does.

In the literature on multiagent systems, the \emph{autonomy} of an agent
(one of the central features distinguishing multiagent systems from other distributed
systems) is sometimes equated with pure selfishness. Under such an interpretation
of the agent paradigm, our notion of equitability would, of course, make little sense.
We believe, however, that it is useful to distinguish different degrees
of autonomy. An agent may well be autonomous in its decision in general, but still
be required to follow certain rules imposed by society (and agreed to by the agent on
entering that society).

\subsection{Local Actions and their Global Effects}

We are now going to prove two lemmas that provide the connection between the local
acceptability criterion given by the notion of equitability and the two egalitarian
social welfare orderings introduced in Section~\ref{sec:swos} (\ie\
the maximin-ordering induced by $sw_e$ as well as the leximin-ordering).

The first lemma shows how global changes are reflected locally.
If a deal happens to increase (global) egalitarian social welfare, that is, if it
results in a rise with respect to the maximin-ordering, then that deal
will in fact be an equitable deal.

\begin{lemma}[Maximin-rise implies equitability]
\label{lem:eswequ}
If $A$ and $A'$ are allocations with $sw_e(A)<sw_e(A')$,
then $\delta=(A,A')$ is an equitable deal.
\end{lemma}

\begin{proof}
Let $A$ and $A'$ be allocations with $sw_e(A)<sw_e(A')$ and let
$\delta=(A,A')$.
Any agent with minimal utility for allocation $A$ must be involved in
$\delta$, because egalitarian social welfare, and thereby these agents'
individual utility, is higher for allocation $A'$.
That is, we have $\min\{u_i(A)\;|\;i\in{\cal A}^\delta\}=sw_e(A)$.
Furthermore, because ${\cal A}^\delta\subseteq{\cal A}$, we certainly have
$sw_e(A')\leq\min\{u_i(A')\;|\;i\in{\cal A}^\delta\}$. 
Given our original assumption of $sw_e(A)<sw_e(A')$, we now 
obtain the inequation
$\min\{u_i(A)\;|\;i\in{\cal A}^\delta\}<\min\{u_i(A')\;|\;i\in{\cal A}^\delta\}$.
This shows that $\delta$ will indeed be an equitable deal.
\end{proof}

\afterproof
Observe that the converse does not hold; not every equitable deal will necessarily
increase egalitarian social welfare (although an equitable deal will never decrease 
egalitarian social welfare either). 
This is, for instance, not the case if only
agents who are currently better off are involved in a deal.
In fact, as argued already at the end of Section~\ref{sec:dealclasses}, 
there can be no class of deals characterisable by a \emph{local}
rationality criterion (see Definition~\ref{def:locality}) that would 
always result in an increase in egalitarian social welfare.

To be able to detect changes in welfare resulting from an equitable deal we
require the finer differentiation between alternative allocations of resources given
by the leximin-ordering. In fact, as we shall see next, any equitable deal can
be shown to result in a strict improvement with respect to the leximin-ordering.

\begin{lemma}[Equitability implies leximin-rise]
\label{lem:equleximin}
If $\delta=(A,A')$ is an 
equitable deal, 
then $A\prec A'$.
\end{lemma}

\begin{proof}
Let $\delta=(A,A')$ be a deal that satisfies the equitability criterion 
and define $\alpha=\min\{u_i(A)\;|\;i\in{\cal A}^\delta\}$.
The value $\alpha$ may be considered as partitioning the ordered utility vector $\vec{u}(A)$
into three subvectors: Firstly, $\vec{u}(A)$ has got a (possibly empty)
prefix $\vec{u}(A)^{<\alpha}$ where all elements are strictly lower than $\alpha$.
In the middle, it has got a subvector $\vec{u}(A)^{=\alpha}$ (with at least one element)
where all elements are equal to $\alpha$. Finally, $\vec{u}(A)$ has got a suffix
$\vec{u}(A)^{>\alpha}$ (which again may be empty) where all elements are strictly greater
than $\alpha$.

By definition of $\alpha$, the deal $\delta$ cannot affect agents
whose utility values belong to $\vec{u}(A)^{<\alpha}$. Furthermore, by definition
of equitability, we have $\alpha<\min\{u_i(A')\;|\;i\in{\cal A}^\delta\}$, which
means that all of the agents that \emph{are}
involved will end up with a utility value which is strictly greater than $\alpha$,
and at least one of these agents will come from $\vec{u}(A)^{=\alpha}$.
We now collect the information we have on $\vec{u}(A')$, the ordered utility vector
of the second allocation $A'$. 
Firstly, it will have a prefix $\vec{u}(A')^{<\alpha}$ identical to 
$\vec{u}(A)^{<\alpha}$. This will be followed by a (possibly empty) subvector 
$\vec{u}(A')^{=\alpha}$ where all elements
are equal to $\alpha$ and which must be strictly shorter than $\vec{u}(A)^{=\alpha}$.
All of the remaining elements of $\vec{u}(A')$ will be strictly greater than $\alpha$.
It follows that $\vec{u}(A)$ lexicographically precedes $\vec{u}(A')$, \ie\
$A\prec A'$ holds as claimed.
\end{proof}

\afterproof
Again, the converse does not hold, \ie\
not every deal resulting in a leximin-rise is necessarily equitable.
Counterexamples are deals where the utility value of the weakest agent involved
stays constant, despite there being an improvement with respect to the leximin-ordering
at the level of society.

A well-known result in welfare economics states that every Pigou-Dalton utility
transfer results in a leximin-rise \cite{Moulin1988a}. Given that we have observed
earlier that every deal that amounts to a Pigou-Dalton transfer will also be an
equitable deal, this result can now also be regarded as a simple corollary to
Lemma~\ref{lem:equleximin}.

\subsection{Maximising Egalitarian Social Welfare}

Our next aim is to prove a convergence result for the egalitarian framework
(in analogy to Theorems~\ref{thm:sufficientwithmoney} and~\ref{thm:sufficientwithoutmoney}). 
We are going to show that systems where agents negotiate equitable deals always converge 
towards an allocation 
with maximal egalitarian social welfare.

\begin{theorem}[Maximal egalitarian social welfare]\label{thm:eswmax}
Any sequence of equitable deals will eventually result in an
allocation of resources with maximal egalitarian social welfare.
\end{theorem}

\begin{proof}
By Lemma~\ref{lem:equleximin}, any equitable deal will result
in a strict rise with respect to the leximin-ordering $\prec$
(which is both irreflexive and transitive).
Hence, as there are only a finite number of distinct allocations,
negotiation will have to terminate after a finite number of deals.
So suppose negotiation has terminated and no more equitable deals are
possible. Let $A$ be the corresponding terminal allocation of resources.
The claim is that $A$ will be an allocation with maximal egalitarian social
welfare. For the sake of contradiction, assume it is not, \ie\
assume there exists another allocation $A'$ for the same system
such that $sw_e(A)<sw_e(A')$. But then, by Lemma~\ref{lem:eswequ},
the deal $\delta=(A,A')$ will be an equitable deal. Hence, there is
still a possible deal, namely $\delta$, which contradicts our earlier
assumption of $A$ being a terminal allocation.
This shows that $A$ will be an allocation with maximal egalitarian social
welfare, which proves our claim.
\end{proof}

\afterproof
From a purely practical point of view, Theorem~\ref{thm:eswmax} may be of 
a lesser interest than the corresponding results for utilitarian systems, 
because it does not refer to an acceptability criterion that only depends 
on a \emph{single} agent.
Of course, this coincides with our intuitions about egalitarian societies:
maximising social welfare is only possible by means of cooperation and the
sharing of information on agents' preferences.

After having reached the allocation with maximal egalitarian social
welfare, it may be the case that still some equitable deals are possible,
although they would not increase social welfare any further (but they would
still cause a leximin-rise).
This can be demonstrated by means of a simple example.
Consider a system with three agents and two resources.
The following table fixes the utility functions:

\begin{center}
\begin{tabular}{lll} \hline
$u_1(\emptyset) = 0$   & $u_2(\emptyset) = 6$     & $u_3(\emptyset) = 8$ \\
$u_1(\{r_1\}) = 5$     & $u_2(\{r_1\}) = 7$       & $u_3(\{r_1\}) = 9$ \\ 
$u_1(\{r_2\}) = 0$     & $u_2(\{r_2\}) = 6.5$     & $u_3(\{r_2\}) = 8.5$ \\ 
$u_1(\{r_1,r_2\}) = 5$ & $u_2(\{r_1,r_2\}) = 7.5$ & $u_3(\{r_1,r_2\}) = 9.5$ \\ 
\hline
\end{tabular}
\end{center} 
A possible interpretation of these functions would be the following.
Agent~3 is fairly well off in any case; obtaining either of the resources $r_1$ 
and $r_2$ will not have a great impact on its personal welfare. The same is
true for agent~2, although it is slightly less well off to begin with.
Agent~1 is the poorest agent and attaches great value to $r_1$, but has no
interest in $r_2$. Suppose agent~3 initially holds both resources.
This corresponds to the ordered utility vector $\langle 0,6,9.5\rangle$.
Passing $r_1$ to agent~1 would lead to a new allocation with the ordered
utility vector $\langle 5,6,8.5\rangle$ and increase egalitarian social 
welfare to $5$, which is the maximum that is achievable in this 
system.  However, there is still another equitable deal that could be implemented
from this latter allocation: agent~3 could offer $r_2$ to agent~2. Of course, 
this deal does not affect agent~1. The resulting allocation would then have 
the ordered utility vector $\langle 5,6.5,8\rangle$, which corresponds to the 
leximin-maximal allocation.

To be able to detect situations where a social welfare maximum has already been 
reached but some equitable deals are still possible, and to be able to stop
negotiation (assuming we are only interested in maximising $sw_e$ as quickly
as possible), however, we would require a \emph{non-local} rationality criterion.
No criterion that only takes the welfare of agents involved in a 
particular deal into account could be sharp enough to always tell us
whether a given deal would increase the minimum utility 
in society (see also our discussion after Lemma~\ref{lem:eswequ}).
We could define a class of \emph{strongly equitable} deals that are like equitable
deals but on top of that require the (currently) weakest agent to be involved in
the deal. This would be a sharper criterion, but it would also be against the spirit
of distributivity and locality, because every single agent would be involved
in every single deal (in the sense of everyone having to announce their utility
in order to be able to determine who is the weakest).

\subsection{Necessity Result}

As our next theorem will show, if we restrict the set of admissible deals
to those that are equitable, then every single deal $\delta$ (that is not 
independently decomposable) may be necessary
to guarantee an optimal result (that is, no sequence of equitable deals excluding
$\delta$ could possibly result in an allocation with maximal egalitarian social
welfare). Furthermore, our theorem improves upon a previous result \cite{EndrissEtAlMFI2003} 
by showing that this holds even when all utility functions are required to be dichotomous.\footnote{%
This theorem also corrects a mistake in the original statement of
the result \cite{EndrissEtAlMFI2003}, where the restriction to deals that are not
independently decomposable had been omitted.}

\begin{theorem}[Necessary deals in egalitarian systems]
\label{thm:egalitariannecessity}
Let the sets of agents and resources be fixed.
Then for every deal $\delta$ that is not independently decomposable, there exist utility functions 
and an initial allocation such that any sequence of equitable deals leading to an allocation 
with maximal egalitarian social welfare would have to include $\delta$.
This continues to be the case even when all utility functions are required to be dichotomous.
\end{theorem}

\begin{proof}
Given a set of agents $\cal A$ and a set of resources $\cal R$, 
let $\delta=(A,A')$ be any deal for this system.
As we have $A\not=A'$, there will be a (at least one) agent $j\in{\cal A}$
with $A(j)\not=A'(j)$. We use this particular $j$ to fix suitable (dichotomous) 
utility functions $u_i$ for agents $i\in{\cal A}$ and sets of resources
$R\subseteq{\cal R}$ as follows:
\[\begin{array}{rcl}
u_i(R) & = & \left\{\begin{array}{ll}
1 & \mbox{if}\ R=A'(i)\ \mbox{or}\ (R=A(i)\ \mbox{and}\ i\not=j) \\
0 & \mbox{otherwise}
\end{array}\right.
\end{array}\]
That is, for allocation $A'$ every agent assigns a utility value of $1$ to the
resources it holds. The same is true for allocation $A$, with the sole exception
of agent~$j$, who only assigns a value of $0$. For any other allocation, agents
assign the value of $0$ to their set of resources, unless that set is the same
as for either allocation $A$ or $A'$. As $\delta$ is not independently decomposable, this
will happen for at least one agent for every allocation different from both $A$ and $A'$.
Hence, for every such allocation at least one agent will assign a utility value of $0$ to its allocated 
bundle. We get $sw_e(A')=1$, $sw_e(A)=0$, and $sw_e(B)=0$ for every other
allocation $B$, \ie\
$A'$ is the only allocation with maximal egalitarian social welfare.

The ordered utility vector of $A'$ is of the form $\langle 1,\ldots,1\rangle$,
that of $A$ is of the form $\langle 0,1,\ldots,1\rangle$, and that of any other
allocation has got the form $\langle 0,\ldots\rangle$, \ie\
we have $A\prec A'$ and $B\preceq A$ for all allocations $B$ with $B\not= A$ and $B\not= A'$.
Therefore, if we make $A$ the initial allocation of resources, then $\delta$ will
be the only deal that would result in a strict rise with respect to the leximin-ordering.
Thus, by Lemma~\ref{lem:equleximin}, $\delta$ would also be
the only equitable deal. Hence, if the set of admissible deals is restricted
to equitable deals then $\delta$ is indeed necessary to reach an allocation
with maximal egalitarian social welfare.
\end{proof}
 
\afterproof
This result shows, again, that there can be no structurally simple class of deals (such as 
the class of deals only involving two agents at a time) that would be sufficient 
to guarantee an optimal outcome of negotiation. This is the case even when
agents only have very limited options for modelling their preferences (as is the
case for dichotomous utility functions).\footnote{%
However, observe that unlike for our two variants of the framework of rational negotiation,
we do not have a necessity result for scenarios with monotonic utility functions
(see Theorems~\ref{thm:necessarywithmoney} and~\ref{thm:necessarywithoutmoney}).
Using a collection of monotonic utility functions as in the proof of 
Theorem~\ref{thm:necessarywithmoney} would not allow us to draw any conclusions
regarding the respective levels of egalitarian social welfare of $A$ and $A'$
on the one hand, and other allocations $B$ on the other.}

While this negative necessity result is shared with the two other instances
of our negotiation framework we have considered, there are currently no 
positive results on the sufficiency of 1-deals for restricted domains in the 
egalitarian setting (see Theorems~\ref{thm:additive} and~\ref{thm:zeroone}). 
For instance, it is not difficult to construct
counterexamples that show that even when all agents are using additive 0-1
functions, complex deals involving all agents at the same time may be required
to reach an allocation with maximal egalitarian social welfare
\cite<a concrete example may be found in>{EndrissEtAlMFI2003}.

\section{Negotiating Lorenz Optimal Allocations}
\label{sec:lorenz}

In this section, we are going to analyse our framework of resource allocation by negotiation
in view of the notion of \emph{Lorenz optimal} allocations introduced in Definition~\ref{def:lorenz}.
We begin with a somewhat more general discussion of possible local rationality criteria
for the acceptability of a given deal.

\subsection{Local Rationality Criteria and Separability}

So far, we have studied three different variants of our negotiation framework:
$(i)$~rational negotiation with side payments (aiming at maximising utilitarian social welfare);
$(ii)$~rational negotiation without side payments (aiming at Pareto optimal outcomes); and
$(iii)$~negotiation in egalitarian agent societies. 
The first two instances of our framework, where agents are either individually rational or
cooperatively rational, have been natural choices as they formalise the widely made
assumptions that agents are both purely self-interested and myopic (for scenarios
with and without monetary side payments, respectively). 

The third variant of the framework, which applies to egalitarian agent societies, 
is attractive for both conceptual and technical reasons. Conceptually, egalitarian 
social welfare is of interest, because it has largely been neglected in the multiagent 
systems literature  despite being the classical counterpart to the widely used notion 
of utilitarian social welfare.
Technically, the analysis of egalitarian agent societies has been interesting, because
egalitarian social welfare does not admit the definition of a local rationality 
criterion that directly captures the class of deals resulting in an increase 
with respect to this metric. This is why we took the detour via the 
leximin-ordering to prove termination.

The class of social welfare orderings that can be captured by deals conforming
to a local rationality criterion is closely related to the class of \emph{separable}
social welfare orderings \cite{Moulin1988a}. In a nutshell, a social welfare
ordering is separable iff it only depends on the agents \emph{changing utility}
whether or not a given deal will result in an increase in social welfare.
Compare this with the notion of a local rationality criterion (Definition~\ref{def:locality});
here it only depends on the agents \emph{changing bundle} whether or not a deal is
acceptable. A change in utility presupposes a change in bundle (but not vice versa).
Hence, every  separable social welfare ordering corresponds to a class of deals 
characterised by a local rationality criterion (but not vice versa).
This means that every separable social welfare ordering gives rise to a local
rationality criterion and proving a general convergence theorem becomes 
straightforward.\footnote{Indeed, in the case of the framework of rational 
negotiation with side payments, the central argument in the proof of 
Theorem~\ref{thm:sufficientwithmoney} has been Lemma~\ref{lem:rationaldeals},
which shows that individual rationality is in fact equivalent to the local 
rationality criterion induced by $sw_u$.}
The leximin-ordering, for instance, is separable \cite{Moulin1988a},
which is why we are not going to further discuss this social welfare ordering
in this paper.\footnote{A suitable rationality criterion would simply amount
to a lexicographic comparison of the ordered utility vectors for the subsociety
of the agents involved in the deal in question.}

Similarly, the ordering over alternative allocations induced by the notion
of Lorenz domination (Definition~\ref{def:lorenz}) is also separable.
Hence, the definition of an appropriate class of deals and the proof
of a general convergence result would not yield any significant new insights either.
However, here the analysis of the effects that some of the previously introduced
rationality criteria have on agent societies when social well-being is 
assessed in terms of the Lorenz condition \emph{is} rather  instructive, 
as we are going to see next.

\subsection{Lorenz Domination and Existing Rationality Criteria}

We are now going to try to establish connections between the global welfare
measure induced by the notion of Lorenz domination on the one hand, and various
local criteria on the acceptability of a proposed deal that individual agents
may choose to apply on the other. For instance, it is an immediate
consequence of Definitions~\ref{def:lorenz} and~\ref{def:cooperativerationality} 
that, whenever $\delta=(A,A')$ is a cooperatively rational deal,
then $A$ must be Lorenz dominated by $A'$.
As may easily be verified, any deal that amounts to a Pigou-Dalton transfer 
(see Definition~\ref{def:pigoudalton}) will also result in a Lorenz improvement.
On the other hand, it is not difficult to construct examples that show
that this is not the case for the class of equitable deals anymore (see
Definition~\ref{def:equitable}). That is, while some equitable 
deals will result in a Lorenz improvement, others will not.

Our next goal is to check whether it is possible to combine existing rationality 
criteria and define a class of deals that captures the notion of
Lorenz improvements in as so far as, for any two allocations $A$ and $A'$
such that $A$ is Lorenz dominated by $A'$, there exists a sequence of deals
(or possibly even a single deal) belonging to that class leading from $A$ to $A'$.
Given that both cooperatively rational deals and Pigou-Dalton transfers
always result in a Lorenz improvement, the union of these two classes of deals
may seem like a promising candidate. In fact, according to a result reported by
\citeA[Lemma~2.3]{Moulin1988a}, it is the case that any Lorenz improvement
can be implemented by means of a sequence of Pareto improvements (\ie\
cooperatively rational exchanges) and Pigou-Dalton transfers.
It is important to stress that this seemingly general result does
\emph{not} apply to our negotiation framework.
To see this, we consider the following example:

\begin{center}
\begin{tabular}{lll} \hline
$u_1(\emptyset) = 0$   & $u_2(\emptyset) = 0$   & $u_3(\emptyset) = 0$ \\
$u_1(\{r_1\}) = 6$     & $u_2(\{r_1\}) = 1$     & $u_3(\{r_1\}) = 1$ \\ 
$u_1(\{r_2\}) = 1$     & $u_2(\{r_2\}) = 6$     & $u_3(\{r_2\}) = 1$ \\ 
$u_1(\{r_1,r_2\}) = 7$ & $u_2(\{r_1,r_2\}) = 7$ & $u_3(\{r_1,r_2\}) = 10$\\ 
\hline
\end{tabular}
\end{center} 
Let $A$ be the allocation in which agent~3 owns both resources, \ie\
$\vec{u}(A)=\langle 0,0,10\rangle$ and
utilitarian social welfare is currently $10$. Allocation $A$ is Pareto optimal,
because any other allocation would be strictly worse for agent~$3$. Hence, there can
be no cooperatively rational deal that would be applicable in this situation.
We also observe that any deal involving only two agents would at best result in a new
allocation with a utilitarian social welfare of $7$ (this would be a deal consisting either
of passing both resources on to one of the other agents, or of passing the ``preferred'' resource
to either agent~$1$ or agent~$2$, respectively). Hence, no deal involving only two agents
(and in particular no Pigou-Dalton transfer) could possibly result in a Lorenz improvement.
However, there \emph{is} an allocation that Lorenz dominates $A$, namely the allocation
assigning to each one of the first two agents their respectively preferred resource.
This allocation $A'$ with $A'(1)=\{r_1\}$, $A'(2)=\{r_2\}$ and $A'(3)=\emptyset$ has
got the ordered utility vector $\langle 0,6,6\rangle$.

The reason why the general result reported by Moulin is not applicable to our domain is 
that we cannot use Pigou-Dalton transfers to implement arbitrary utility transfers here.
Moulin assumes that every possible utility vector constitutes a feasible agreement.
In the context of resource allocation, this would mean that there is an allocation
for every possible utility vector.
In our framework, where agents negotiate over a finite number of indivisible
resources, however, the range of feasible allocations is limited. 
For instance, in the above example there is no feasible allocation with the 
(ordered) utility vector $\langle 0,4,6\rangle$. In Moulin's system, agents could first
move from $\langle 0,0,10\rangle$ to  $\langle 0,4,6\rangle$ (a Pigou-Dalton transfer) and 
then from $\langle 0,4,6\rangle$ to $\langle 0,6,6\rangle$ (a Pareto improvement).
In our system, on the other hand, this is not possible.

\subsection{Simple Pareto-Pigou-Dalton Deals and 0-1 Scenarios}

As we cannot use existing rationality criteria to compose a criterion
that captures the notion of a Lorenz improvement (and that would allow us
to prove a general convergence theorem), we are going to investigate
how far we can get in a scenario with restricted utility functions.
Recall our definition of 0-1 scenarios where utility functions can only be used to indicate 
whether an agent does or does not need a particular resource.
As we shall see next, for 0-1 scenarios, the aforementioned result of Moulin \emph{does}
apply. In fact, we can even sharpen it a little by showing that only Pigou-Dalton transfers
and cooperatively rational deals involving just a single resource and two agents each are
required to guarantee negotiation outcomes that are Lorenz optimal.
We first give a formal definition of this class of deals.
\begin{definition}[Simple Pareto-Pigou-Dalton deals]
A deal $\delta$ is called a simple Pareto-Pigou-Dalton deal iff 
it is a 1-deal and either cooperatively rational or 
a Pigou-Dalton transfer.
\end{definition}
We are now going to show that this class of deals is sufficient to guarantee
Lorenz optimal outcomes of negotiations in 0-1 scenarios

\begin{theorem}[Lorenz optimal outcomes]\label{thm:lorenz}
In 0-1 scenarios, any sequence of simple Pareto-Pigou-Dalton deals will eventually
result in a Lorenz optimal allocation of resources.
\end{theorem}

\begin{proof}
As pointed out earlier, any deal that is either cooperatively rational or a Pigou-Dalton
transfer will result in a Lorenz improvement (not only in the case of 0-1 scenarios).
Hence, given that there are only a finite number of distinct allocations of resources, 
after a finite number of deals the system will have reached an allocation where no more 
simple Pareto-Pigou-Dalton deals are possible; that is, negotiation must terminate.
Now, for the sake of contradiction, let us assume this
terminal allocation $A$ is not optimal, \ie\
there exists another allocation $A'$ that Lorenz dominates $A$.
Amongst other things, this implies $sw_u(A)\leq sw_u(A')$, \ie\
we can distinguish two cases: either $(i)$~there has been a strict increase in
utilitarian welfare, or $(ii)$~it has remained  constant. In 0-1 scenarios, the former
is only possible if there are (at least) one resource $r\in{\cal R}$ and two
agents $i, j\in{\cal A}$ such that $u_i(\{r\})=0$ and $u_j(\{r\})=1$ as well as
$r\in A(i)$ and $r\in A'(j)$, \ie\
$r$ has been moved from agent $i$ (who does not need it) to agent $j$ (who
does need it). But then the 1-deal of moving only $r$ from 
$i$ to $j$ would be cooperatively rational and hence also a simple Pareto-Pigou-Dalton 
deal. This contradicts our assumption of $A$ being a terminal allocation.

Now let us assume that utilitarian social welfare remained constant, \ie\
$sw_u(A)=sw_u(A')$. Let $k$ be the smallest index such that
$\vec{u}_k(A)<\vec{u}_k(A')$. (This is the first $k$ for which the inequality
in Definition~\ref{def:lorenz} is strict.)
Observe that we cannot have $k=|{\cal A}|$, as this would contradict
$sw_u(A)=sw_u(A')$. We shall call the agents contributing the first $k$ entries
in the ordered utility vector $\vec{u}(A)$ the \emph{poor} agents and the
remaining ones the \emph{rich} agents.
Then, in a 0-1 scenario, there must be a resource $r\in{\cal R}$ that is
owned by a rich agent $i$ in allocation $A$ and by a poor agent $j$ in allocation $A'$
and that is needed by both these agents, \ie\
$u_i(\{r\})=1$ and $u_j(\{r\})=1$. But then moving this resource from agent~$i$ to agent~$j$
would constitute a Pigou-Dalton transfer (and hence also a simple Pareto-Pigou-Dalton deal)
in allocation $A$, which again contradicts our earlier assumption of $A$ being terminal.
\end{proof}

\afterproof
In summary, we have shown that $(i)$~any allocation of resources from which no simple
Pareto-Pigou-Dalton deals are possible must be a Lorenz optimal allocation and $(ii)$~that
such an allocation will always be reached by implementing a finite number of simple
Pareto-Pigou-Dalton deals. As with our earlier convergence results, agents do not
need to worry about which deals to implement, as long as they are simple Pareto-Pigou-Dalton
deals. The convergence to a global optimum is guaranteed by the theorem.

\section{Further Variations}
\label{sec:further}

In this section, we are going to briefly consider two further notions of social preference
and discuss them in the context of our framework of resource allocation by negotiation. 
Firstly, we are going to introduce the idea of \emph{elitist} agent societies, where social welfare 
is tied to the welfare of the agent that is currently best off. Then we are going to discuss 
societies where \emph{envy-free} allocations of resources are desirable.

\subsection{Elitist Agent Societies}
\label{sec:elitist}

Earlier we have discussed the maximin-ordering 
induced by the  egalitarian collective utility function $sw_e$ (see Definition~\ref{def:swe}). 
This ordering is actually a particular case of a class of social welfare orderings,
sometimes called \emph{k-rank dictators} \cite{Moulin1988a}, where a particular agent 
(the one corresponding to the $k$th element in the ordered utility vector) is chosen
to be the representative of society. 
Amongst this class of orderings, another particularly interesting case is where 
the welfare of society is evaluated on the basis of the happiest agent (as opposed 
to the unhappiest agent, as in the case for egalitarian welfare). We call this the
\emph{elitist} approach to measuring social welfare.

\begin{definition}[Elitist social welfare]
The elitist social welfare $sw_{el}(A)$ of an allocation of resources $A$ is
defined as follows:
\[\begin{array}{rcl}
sw_{el}(A) & = & \max\{u_i(A)\;|\;i\in{\cal A}\}
\end{array}\]
\end{definition}
In an elitist agent society, agents would cooperate in order
to support their champion (the currently happiest agent).
While such an approach 
may seem somewhat unethical as far as human society is concerned,
we believe that it could indeed be very appropriate for certain societies of
artificial agents. For some applications, a distributed multiagent system 
may merely serve as a means for helping a single agent in that system to achieve 
its goal. However, it may not always be known in advance which agent is most 
likely to achieve its goal and should therefore be supported by its peers. 
A typical scenario could be where a system designer launches
different agents with the same goal, with the aim that \emph{at least one}
agent achieves that goal ---no matter what happens to the others.
As with egalitarian agent societies, this does not contradict the idea
of agents being \emph{autonomous} entities. Agents may be physically distributed
and make their own autonomous decisions on a variety of issues whilst also adhering
to certain social principles, in this case elitist ones.

From a technical point of view, designing a criterion that would allow
agents inhabiting an elitist agent society to decide locally whether or
not to accept a particular deal is very similar to the egalitarian case. In analogy
to the case of equitable deals defined earlier, a suitable deal would have to
increase the maximal individual welfare amongst the agents involved in any one deal.
As for the egalitarian case, there can be no class of deals characterised by a
local rationality criterion that would \emph{exactly} capture the range of deals
resulting in an increase in elitist social welfare (because every agent in the system would
have to be consulted first to determine who is currently best off). To prove convergence, 
we would have to resort to an auxiliary social welfare ordering (similarly to the
use of the leximin-ordering in the proof of Theorem~\ref{thm:eswmax}).

Of course, in many cases there is a much simpler way of finding
an allocation with maximal elitist social welfare. For instance, if all agents
use monotonic utility functions, then moving all resources to the agent
assigning the highest utility value to the full bundle $\cal R$ would 
be optimal from an elitist point of view. More generally, we can 
always find an elitist optimum by checking whose utility function 
has got the highest peak.
That is, while the highest possible elitist social welfare can easily be 
determined in a centralised manner, our distributed approach can 
still provide a useful framework for studying the process of 
actually reaching such an optimal allocation.

\subsection{Reducing Envy amongst Agents}
\label{sec:envy}

Our final example for an interesting approach to measuring social welfare in
an agent society is the issue of \emph{envy-freeness} \cite{BramsTaylor1996a}.
For a particular allocation of resources, an agent may be ``envious'' of another
agent if it would prefer that agent's set of resources over its own.
Ideally, an allocation should be envy-free.
\begin{definition}[Envy-freeness]
An allocation of resources $A$ is called envy-free iff we have
$u_i(A(i))\geq u_i(A(j))$ for all agents $i, j\in{\cal A}$.
\end{definition}
Like egalitarian social welfare, this is related to the fair division of resources amongst agents.
Envy-freeness is desirable (though not always achievable) in societies of
self-interested agents in cases where agents have to collaborate with each other
over a longer period of time. In such a case, should an agent believe
that it has been ripped off, it would have an incentive to leave the coalition
which may be disadvantageous for other agents or the society as a whole.
In other words, envy-freeness plays an important role with respect to the
stability of a group. Unfortunately, envy-free allocations do not always exist.
A simple example would be a system with two agents and just a single resource,
which is valued by both of them. Then whichever agent holds that single resource
will be envied by the other agent.

Furthermore,
aiming at agreeing on an envy-free allocation of resources is not always compatible 
with, say, negotiating Pareto optimal outcomes. Consider the following example of 
two agents with identical preferences over alternative bundles of resources:

\begin{center}
\begin{tabular}{ll} \hline
$u_1(\emptyset) = 0$ & $u_2(\emptyset) = 0$ \\
$u_1(\{r_1\}) = 1$ & $u_2(\{r_1\}) = 1$ \\ 
$u_1(\{r_2\}) = 2$ & $u_2(\{r_2\}) = 2$ \\ 
$u_1(\{r_1,r_2\}) = 0$ & $u_2(\{r_1,r_2\}) = 0$ \\ \hline 
\end{tabular}
\end{center}
For this example, either one of the two allocations where one agent owns all
resources and the other none would be envy-free (as no agent would prefer
the other one's bundle over its own). However, such an allocation would not
be Pareto optimal. On the other hand, an allocation where each agent owns
a single resource would be Pareto optimal, but not envy-free (because the 
agent holding $r_1$ would rather have $r_2$).

We should stress that envy is defined on the 
sole basis of an agent's private preferences, \ie\ 
there is no need to take other
agents' utility functions into account. Still, whether an agent is envious
does not just depend on the resources it holds, but also on the resources it
\emph{could} hold and whether any of the other agents currently hold a preferred bundle.
This somewhat paradoxical situation makes envy-freeness far less amenable to our 
methodology than any of the other notions of social welfare we have discussed in 
this paper.

To be able to measure different \emph{degrees} of envy, we could, for example,
count the number of agents that are envious for a given allocation.
Another option would be to compute for each agent $i$ that experiences any envy at all
the difference between $u_i(A(i))$ and $u_i(A(j))$ for the agent $j$ that $i$ envies 
the most. Then the sum over all these differences would also provide an indication
of the degree of overall envy (and thereby of social welfare).
In the spirit of egalitarianism, a third option would be identify the degree of envy 
in society with the degree of envy experienced by the agent that is the most 
envious \cite{LiptonEtAlEC2004}.
However, it is not possible to define a local acceptability criterion
in terms of the utility functions of the agents involved in a deal (and only those)
that indicates whether the deal in question would reduce envy according to any such
a metric. This is a simple consequence of the fact that a deal may affect the degree
of envy experienced by an agent not involved in the deal at all (because it could
lead to one of the participating agents ending up with a bundle preferred by the
non-concerned agent in question).

\section{Conclusion}
\label{sec:conclusion}

We have studied an abstract negotiation framework where members of an agent society arrange 
multilateral deals to exchange bundles of indivisible resources, and we have analysed how the 
resulting changes in resource distribution affect society  with respect to different 
social welfare orderings.
 
For scenarios where agents act rationally in the sense of never accepting a deal 
that would (even temporarily) decrease their level of welfare, we have seen that 
systems where side payments are possible can guarantee outcomes with maximal 
utilitarian social welfare, while systems without side payments allow, at least, 
for the negotiation of Pareto optimal allocations. We have also considered two
examples of special domains with restricted utility functions, namely
additive and 0-1 scenarios. In both cases, we have been able to
prove the convergence to a socially optimal allocation of resources also
for negotiation protocols that allow only for deals involving only 
a single resources and a pair of agents each (so-called 1-deals).
In the case of agent societies where welfare is measured in terms of the egalitarian
collective utility function, we have put forward the class of equitable deals and
shown that negotiation processes where agents use equitability as an acceptability 
criterion will also converge towards an optimal state. Another result states that, 
for the relatively simple 0-1 scenarios, Lorenz optimal allocations can be achieved 
using one-to-one negotiation by implementing 1-deals that are either inequality-reducing 
or that increase the welfare of both agents involved.
We have also discussed the case of elitist agent societies where social welfare is
tied to the welfare of the most successful agent. And finally, we have pointed 
out some of the difficulties associated with designing agents that would be able
to negotiate allocations of resources where the degree of envy between the agents
in a society is minimal.
Specifically, we have proved the following technical results:
\begin{itemize}
\item The class of \emph{individually rational} deals\footnote{%
Recall that individually rational deals 
may include monetary side payments.}
is sufficient to negotiate allocations with \emph{maximal utilitarian
social welfare} (Theorem~\ref{thm:sufficientwithmoney}).
\item In domains with \emph{additive utility functions}, the class of
\emph{individually rational 1-deals} is sufficient
to negotiate  allocations with \emph{maximal utilitarian
social welfare} (Theorem~\ref{thm:additive}).
\item  The class of \emph{cooperatively rational} deals is
sufficient to negotiate \emph{Pareto optimal} allocations 
(Theorem~\ref{thm:sufficientwithoutmoney}).
\item In domains with \emph{0-1 utility functions}, the class of
\emph{cooperatively rational 1-deals} is sufficient
to negotiate  allocations with \emph{maximal utilitarian
social welfare} (Theorem~\ref{thm:zeroone}).
\item The class of \emph{equitable} deals is sufficient to negotiate
allocations with \emph{maximal egalitarian social welfare}
(Theorem~\ref{thm:eswmax}).
\item In domains with \emph{0-1 utility functions}, the class of 
\emph{simple Pareto-Pigou-Dalton} deals (which are  1-deals) 
is sufficient to negotiate \emph{Lorenz optimal} allocations 
(Theorem~\ref{thm:lorenz}).
\end{itemize}

\noindent
For each of the three convergence results that apply to deals without 
structural restrictions (rather than to 1-deals),
we have also proved corresponding \emph{necessity results}
(Theorems~\ref{thm:necessarywithmoney}, \ref{thm:necessarywithoutmoney}, 
and~\ref{thm:egalitariannecessity}). These theorems show that any given 
deal (defined as a pair of allocations) that is not independently decomposable may be necessary 
to be able to negotiate an optimal allocation of resources (with respect to the chosen
notion of social welfare), if deals are required to conform to the rationality 
criterion in question. As a consequence of these results, no negotiation
protocol that does not allow for the representation of deals involving any 
number of agents and any number of resources could ever enable agents (whose 
behaviour is constrained by our various rationality criteria) to negotiate a
socially optimal allocation in all cases. Rather surprisingly, all three necessity
results continue to apply even when agents can only differentiate between
resource bundles that they would be happy with and those they they would 
not be happy with (using \emph{dichotomous} utility functions). 
Theorems~\ref{thm:necessarywithmoney} and~\ref{thm:necessarywithoutmoney}
also apply in case all agents are required to use \emph{monotonic} utility functions.

A natural question that arises when considering our convergence results concerns 
the complexity of the negotiation framework. How difficult is it for agents
to agree on a deal and how many deals are required before a system converges to an
optimal state? The latter of these questions has recently been addressed
by \citeA{EndrissMaudetJAAMAS2005}. The paper establishes upper bounds on the number of deals 
required to reach any of the optimal allocations of resources referred to 
in the four convergence theorems for the model of rational negotiation (\ie\
Theorems~\ref{thm:sufficientwithmoney}, \ref{thm:additive}, \ref{thm:sufficientwithoutmoney},
and~\ref{thm:zeroone}). It also discusses the different aspects 
of complexity involved at a more general level (such as the distinction between
the \emph{communication complexity} of the system, \ie\
the amount of information that agents need to exchange to reach 
an optimal allocation, and the \emph{computational complexity} 
of the reasoning tasks faced by every single agent).
\citeA{DunneJAIR2005} addresses a related problem and studies the number of deals 
meeting certain structural requirements (in particular 1-deals) 
that are required to reach a given target allocation (whenever this is possible at
all ---recall that our necessity results show that excluding certain deal patterns 
will typically bar agents from reaching optimal allocations).

In earlier work, \citeA{DunneEtAlAIJ} have studied the complexity of deciding 
whether one-resource-at-a-time trading with side payments is sufficient to reach a
given allocation (with improved utilitarian social welfare). 
This problem has been shown to be NP-hard. Other complexity results concern
the computational complexity of finding a socially optimal allocation, 
independently from the concrete negotiation mechanism used. As mentioned earlier,
such results are closely related to the computational complexity of the winner determination
problem in combinatorial auctions \cite{RothkopfEtAlMS1998,CramtonEtAl2006}. 
Recently, NP-completeness results for this optimisation problem have been derived 
with respect to several different ways of representing utility 
functions \cite{DunneEtAlAIJ,ChevaleyreEtAlCSDT2004}.
\citeA{BouveretLangIJCAI2005} also address the computational complexity
of deciding whether an allocation exists that is both envy-free and Pareto optimal.

Besides presenting technical results,
we have argued that a wide spectrum of social welfare orderings (rather than just
those induced by the well-known utilitarian collective welfare function and the 
concept of Pareto optimality) can be of interest to agent-based applications.
In the context of a typical electronic commerce application, where participating
agents have no responsibilities towards each other, a system designer may wish to ensure
Pareto optimality to guarantee that agents get maximal payoff whenever this is possible 
without making any of the other agents worse off. 
In applications where a \emph{fair} treatment of all participants is vital (\eg\
cases where the system infrastructure is jointly owned by all the agents), 
an egalitarian approach to measuring social welfare may be more appropriate.
Many applications are in fact likely to warrant a mixture of utilitarian and
egalitarian principles. Here, systems that enable Lorenz optimal agreements
may turn out to be the technology of choice. Other applications, however, may 
require social welfare to be measured in ways not foreseen by the models typically 
studied in the social sciences. Our proposed notion of elitist welfare
would be such an example. Elitism has little room in human society, where ethical 
considerations are paramount, but for a particular computing application 
these considerations may well be dropped or changed. 

This discussion suggests an approach to multiagent systems design that we
call \emph{welfare engineering} \cite{EndrissMaudetESAW2003}. 
It involves, firstly, the application-driven 
choice (or possibly invention) of a suitable social welfare ordering and,
secondly, the design of agent behaviour profiles and negotiation mechanisms 
that permit (or even guarantee) socially optimal outcomes of interactions 
between the agents in a system. As discussed earlier, designing agent behaviour
profiles does not necessarily contradict the idea of the autonomy of an agent,
because autonomy always has to be understood as being relative to the norms 
governing the society in which the agent operates. 
We should stress that, while we have been studying a \emph{distributed} 
approach to multiagent resource allocation in this paper, the general
idea of exploring the full range of social welfare orderings when 
developing agent-based applications also applies to centralised mechanisms 
(such as combinatorial auctions).

We hope to develop this methodology of welfare 
engineering further in our future work.
Other possible directions of future work include the
identification of further social welfare orderings and the definition 
of corresponding deal acceptability criteria; the continuation of
the complexity-theoretic analysis of our negotiation framework;
and the design of practical trading mechanisms (including both 
protocols and strategies) that would allow agents to agree on 
multilateral deals involving more than just two agents at a time.


\acks{%
We would like to thank J\'{e}r\^{o}me Lang and 
various anonymous referees for their valuable comments. 
This research has been partially supported by the European Commission 
as part of the SOCS project (IST-2001-32530).}


\bibliographystyle{theapa}
\bibliography{deals}

\end{document}